\documentclass[journal,12pt,draftclsnofoot,onecolumn,english]{IEEEtran}
\usepackage{amsmath,bm,amssymb,blkarray,amsthm}
\usepackage{array,color,colortbl}
\usepackage{hyperref}
\usepackage{verbatim}
\usepackage{graphicx}
\usepackage{cases}
\usepackage{setspace}
\usepackage{amsfonts}
\usepackage{indentfirst}
\usepackage{cite}
\usepackage{caption}
\usepackage{algorithm}
\usepackage{algpseudocode}
\usepackage{booktabs}
\usepackage{subfigure}
\usepackage{multirow}
\usepackage{amsthm}
\usepackage{diagbox}
\usepackage{pstricks}
\makeatletter
\newtheoremstyle{mythm}{3pt}{3pt}{}{16pt}{\bfseries}{:}{.5em}{}
\theoremstyle{mythm}
\newtheorem{theorem}{Theorem}
\newtheorem{example}{Example}
\newtheorem{definition}{Definition}
\newtheorem{remark}{Remark}

\newtheorem{lemma}{Lemma}

\hypersetup{
	linkcolor=blue,
	filecolor=blue,
	urlcolor=red,
	citecolor=green,
}

\begin{document}
	\title{Multiple-antenna Placement Delivery Array   for Cache-aided MISO Systems}
\author{Ting Yang,  Kai~Wan,~\IEEEmembership{Member,~IEEE,} Minquan Cheng
and~Giuseppe~Caire,~\IEEEmembership{Fellow,~IEEE}
\thanks{T. Yang and M. Cheng are with Guangxi Key Lab of Multi-source Information Mining $\&$ Security, Guangxi Normal University,
Guilin 541004, China  (e-mail: yt\_yang\_ting@163.com, chengqinshi@hotmail.com). }
\thanks{K. Wan and G. Caire are with the Electrical Engineering and Computer Science Department, Technische Universit\"{a}t Berlin,
10587 Berlin, Germany (e-mail: kai.wan@tu-berlin.de, caire@tu-berlin.de).  The work of K.~Wan and G.~Caire was partially funded by the European Research Council under the ERC Advanced Grant N. 789190, CARENET.}
}
	\maketitle

\begin{abstract}
We consider the cache-aided multiple-input single-output (MISO) broadcast channel, which consists of a server with $L$ antennas and $K$ single-antenna users, where the server contains $N$ files of equal length and each user is equipped with a local cache of size $M$ files. Each user requests an arbitrary file from library. The objective is to design a coded caching scheme based on uncoded placement and one-shot linear delivery, to achieve the maximum sum Degree-of-Freedom (sum-DoF) with low subpacketization.  It was shown in the literature that under the constraint of uncoded placement and one-shot linear delivery, the optimal sum-DoF is  $L+\frac{KM}{N}$. However, previously proposed schemes for this setting incurred either an exponential subpacketization order in $K$, or required specific conditions in the system parameters $L$, $K$, $M$ and $N$.	In this paper, we propose a new combinatorial structure called multiple-antenna placement delivery array (MAPDA). Based on MAPDA and Latin square, the first proposed scheme achieves the optimal sum-DoF $L+\frac{KM}{N}$ with the subpacketization of $K$ when  $\frac{KM}{N}+L=K$. Subsequently, for the general case we propose a transformation approach to construct an MAPDA from any $g$-regular PDA (a class of PDA where each integer in the array occurs $g$ times) for the original shared-link coded caching problem. When the original PDA corresponds to the Maddah-Ali and Niesen coded caching scheme, the resulting scheme under the combinatorial  structure of MAPDA can achieve the optimal sum-DoF $L+\frac{KM}{N}$  with reduced subpacketization  with respect to the existing schemes. The work can be extended to the multiple independent single-antenna transmitters (servers) corresponding to the cache-aided interference channel proposed by Naderializadeh et al. and the scenario of transmitters equipped with multiple antennas.
	\end{abstract}

\section{Introduction}
	Demands on wireless video are growing at an exponential rate in recent years, which causes the traffic congestion over the wireless channels.
	Efficient video delivery over the air becomes a very important problem.
	An efficient solution to  alleviate networks from content-related traffic is content caching, which pre-stores some library content at user devices during off-peak traffic hour; thus the stored content will not be further transmitted.
	Coded caching was originally proposed by   Maddah-Ali and Niesen (MN) in \cite{MN} for the single-input single-output (SISO) shared-link network model, where
	a central server with access to a library containing $N$ files is connected to $K$ cache-aided users through an error-free shared-link and each user has a cache to store at most $M$ files.
	By the transmission of multicast messages and the use of cached content in order to remove interference, coded caching leads to a coded caching gain  in addition to the conventional uncoded caching gain.
	A coded caching scenario contains two phases, namely {\it placement} and {\it delivery}. 	
	In the placement phase, each user stores some packets of each file without knowledge of users' later demands. In the delivery phase,  each user requests one file in the library. According to users' caches and demands, the server broadcasts coded packets to the users such that all demands are satisfied.
	For each  $M = \frac{N t}{K}$ where $t\in \{0,1,\ldots,K-1\}$,
	each coded packet transmitted in the delivery phase is simultaneously useful to $ \frac{KM}{N}+1$ users, where $ \frac{KM}{N}+1$ represents the achieved coded caching gain.
	The subpacketization of a coded caching scheme refers to the number of subfiles per file. Although information theoretic achievability and converse results are typically studied in the limit of very large file size, such that	the subpacketization does not represent a limitation, in practice for finite size files, it is important to design schemes with low subpacketization. The subpacketization of the MN scheme
	grows exponentially with $K$.
	To reduce the subpacketization, various combinatorial subfile assignments have been proposed in the literature, such as the authors in \cite{YCTC} proposed a combination structure  referred to as  placement delivery array (PDA).
	This is an array whose entries are integers and a special symbol $*$ (star), where the positions of the stars indicate which subfiles are cached and the integers indicate which subfiles are jointly encoded into the multicast messages (see later for a formal definition).
	The  MN coded caching scheme can also be represented as a PDA, referred to as MN PDA. Remarkably, various schemes based on PDA proposed in \cite{CJYT,CJWY,CWZW,WCWC} have lower subpacketization than the MN scheme. A PDA is called $g$-regular if each integer appears $g$ times in the array. In the literature the problem of reducing the subpacketization of coded caching has been widely studied and several other combinatorial constructions have been proposed, such as
 	the linear block codes \cite{TR}, the special $(6,3)$-free hypergraphs \cite{SZG}, the $(r,t)$ Ruzsa-Szem\'{e}redi graphs \cite{STD}, the strong edge coloring of bipartite graphs \cite{YTCC}, the projective space \cite{KP} and other combination design \cite{CLTW}.

	Following the seminal works of MN,  coded caching was applied to a variety of network topologies, such as    Device-to-Device (D2D) networks~\cite{JCM}, hierarchical   networks~\cite{KNMD,JWTCEL},  arbitrary multi-server linear networks~\cite{SSB}, etc.	
	Coded caching was   extended to   the wireless interference channel  with multiple single-antenna cache-aided transmitters and receivers~\cite{NMA,HND}, whose objective is to maximize the system sum Degree-of-Freedom (sum-DoF).
	If each transmitter is able to cache the whole library, the problem  reduces  to the cache-aided multiple-input single-output (MISO) broadcast channel (BC) with $L$ antennas  studied in \cite{SCH,EP,SPSET,MB,ST,STSK,PJC}.
	With    one-shot linear coding schemes based on  the joint design of coded caching and zero-forcing (ZF) precoding,  the sum-DoF $L+ \frac{KM}{N}$ was achieved in~\cite{NMA,SCH},
	which yields the MN coded caching gain for $L = 1$. It was proved in~\cite{EBPresolving} that under the constraints of uncoded cache placement and one-shot linear delivery,  the sum-DoF $L+ \frac{KM}{N}$ is optimal.
	A problem of  the cache-aided MISO BC schemes 	in~\cite{NMA,SCH} is that it requires subpacketization
	$\binom{K}{KM/N} \binom{K-KM/N-1}{L-1}$, which is even larger than the MN scheme.
	Various works have health with this subpacketization issue  \cite{EP,SPSET,MB,ST,STSK}.
	For the case where $\frac{K}{L}$ and $\frac{KM}{N}/L$  are both integers, the scheme of \cite{EP} achieves the sum-DoF $L+\frac{KM}{N}$ with subpacketization $\binom{K/L}{KM/(NL)}$.
	Under the constraint  $L\geq \frac{KM}{N}$, the authors in~\cite{SPSET}  utilize a  cyclic cache placement to achieve the sum-DoF  $L+ \frac{KM}{N}$ with subpacketization linear with $K$.
    A summary of the different schemes proposed in the literature, achieving the  sum-DoF $L+ \frac{KM}{N}$ subject to conditions and the corresponding subpacketizations is given in Table~\ref{tab-compare}.
	One should note that for the general case, no
	existing scheme can achieve the same sum-DoF of~\cite{NMA,SCH}
	with generally lower subpacketization.
	
	\paragraph*{Our Contributions}	
	This paper considers the cache-aided MISO BC  problem with one-shot linear delivery in~\cite{NMA}. We extend the PDA structure for the shared-link model in \cite{YCTC} to the MISO BC by using ZF, and propose a novel MISO BC coded caching structure, referred to as multiple-antenna placement delivery array (MAPDA). The MAPDA which is a construction structure for the cache-aided MISO BC problem based on uncoded cache placement and one-shot linear delivery,
	  generalizes the  one-shot linear coding constructions in~\cite{NMA,EP,SPSET,MB,STSK}. 
	We then propose two MAPDA constructions with lower subpacketization than the existing schemes while achieving the maximum sum-DoF $L+ \frac{KM}{N}$:
	\begin{itemize}
		\item For the case where $\frac{KM}{N}+L=K$, we propose an MAPDA construction based on the cyclic cache placement and Latin square, which achieves the maximum sum-DoF with subpacketization equals to $K$.
		\item We provide a non-trivial transformation approach to extend   any given $K_1$-user regular $g$-PDA for shared-link caching model to an $mK_1$-user MAPDA for the cache-aided MISO BC problem, where $m$ is any positive integer and $m\leq L$. The achieved sum-DoF is  $L+m(g-1)$ and the needed subpacketization is linear with $F_1$ which represents the subpacketization of the original regular PDA.
		In addition, by setting the original PDA as an MN PDA, the resulting scheme achieves the maximum sum-DoF  $L+ \frac{KM}{N}$ with a much lower subpacketization
		than~\cite{NMA,SCH}.  Interestingly, it can also cover the caching scheme in~\cite{EP} as a special case, i.e., when $m=L$.
	\end{itemize}
	
	\paragraph*{Paper Organization}	
	The rest of this paper is organized as follows. Section~\ref{sec-model} describes the system model. Section \ref{sec-MAPDA} reviews PDA and introduces the structure of MAPDA. Section \ref{sec-const} proposes two constructions of MAPDA and the gives analysis performance.
	Section \ref{sec-conclu} concludes the paper and some proofs can be found in Section \ref{proof-regular} and Appendices.

	\begin{table}  	
	\caption{Existing schemes with the sum-DoF  $L+ \frac{KM}{N}$, $L$ antennas and memory ratio of $\frac{t}{K}$, for $t\in [K]$.}
	\centering
	\begin{spacing}{1.5}
		\begin{tabular}{|m{3.5cm}<{\centering}|m{4cm}<{\centering}|m{4cm}<{\centering}|}	\hline
			Scheme & Limitation & Subpacketization \\	
			\hline Scheme in \cite{NMA}& No limitations  & $\binom{K}{t}\frac{t!(K-t-1)!}{(K-t-L)!}$\\
				
			\hline Scheme in \cite{SCH}& No limitations & $\binom{K}{t}\binom{K-t-1}{L-1}$\\
				
			\hline Scheme in \cite{EP}& $\frac{K}{L},\frac{t}{L}\in\mathbb Z^+$ & $\binom{K/L}{t/L}$\\
				
			\hline Scheme in \cite{SPSET}&$t\leq L$ & $\frac{K(t+L)}{(gcd(K,t,L))^2}$\\
				
			\hline Scheme in \cite{MB}& $\frac{t+L}{t+1}\in\mathbb{Z}^+$ &$\binom{K}{t}$\\
			\hline	
		\end{tabular}
		\label{tab-compare}
	\end{spacing}
	\end{table}
	
	\paragraph*{Notations}
	\begin{itemize}
		\item $[a:b]=\{a,a+1,\cdots,b\}$ and $[a]=\{1,2,\cdots,a\}$.
		
		\item For any array $\mathbf{P}$ composed of $m$ rows and $n$ columns, the $\mathbf{P}(i,j)$ denotes the entry in the $i^\text{th}$ row and $j^{\text{th}}$	column.	$(i,j)$ is also referred to as the position of $\mathbf{P}(i,j)$ in $\mathbf{P}$.
		\item Given an array $\mathbf{P}$ composed of the symbol $``*"$ and $S$ integers, for any integer $a$, the new array  $\mathbf{P}+a$ denotes integer $a$ plus each entry in $\mathbf{P}$, where $*+a=*$.
		\item $\text{gcd}(a,b)$ represents the greatest common divisor of integers $a$ and $b$.
	\end{itemize}

	\section{System model}\label{sec-model}
	\begin{figure}
		\centering
		\includegraphics[width=6in]{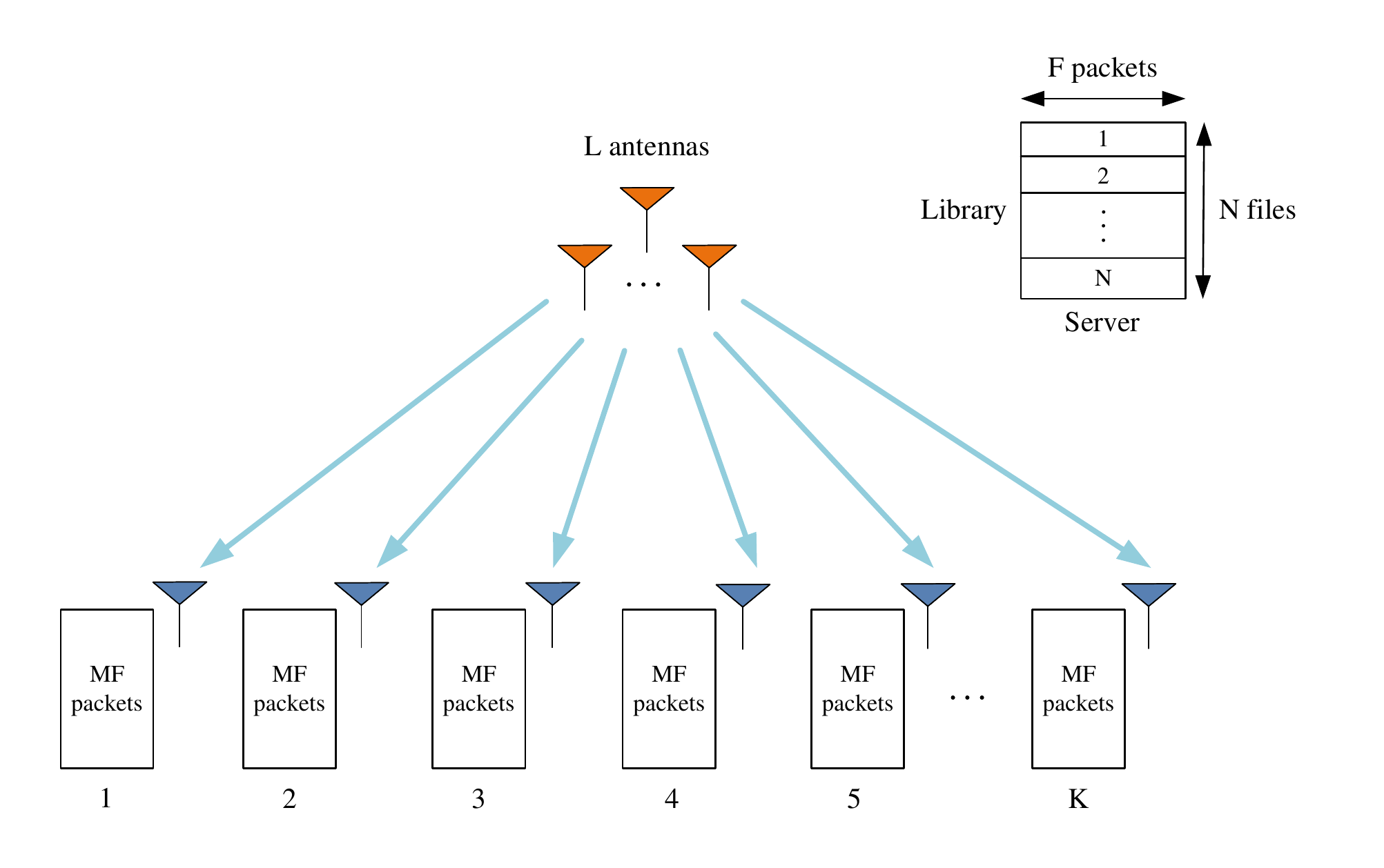}
		\caption{Multiple-input-single-output broadcast channel (MISO) BC system of one server with $L$ antennas and $K$ users with single antenna.}
		\label{fig-model}
	\end{figure}
	This paper considers the $(L,K,M,N)$ cache-aided multiple-input single-output (MISO) broadcast problem with one-shot linear delivery considered in \cite{NMA}, as illustrated in Fig. \ref{fig-model}. A server with $L$ antennas has access to the library containing $N$  files, denoted by $\mathcal{W}=\{\mathbf{W}_n\ |\ n\in [N]\}$.
	Each file $\mathbf{W}_n$ in the library consists of $F$ packets  $\mathbf{W}_n\triangleq \{\mathbf{W}_{n,f}\ |\ f\in[F]\}$, where each packet denoted by $\mathbf{W}_{n,f}\in \mathbb{F}_2^B$ contains $B$ uniformly i.i.d. bits.
	Each user is equipped with one antenna and a cache of $MF$ packets, where $0\leq M\leq N$.
	
	The communication process at time slot $t$ between the server and users can be modelled as
	\begin{equation}
		Y_k(t)=\sum_{i=1}^{L}h_{k,i}X_i(t)+\epsilon_k(t), \label{channel model}
	\end{equation}
	where $X_i(t)\in\mathbb{C}$ denotes the signal sent by antenna $i\in[L]$ which satisfies the power constraint $\mathbb{E}\big[\sum_{i\in[L]}|X_i(t)|^2\big]\leq P$. $Y_k(t)$ denotes the signal received by user $k\in[K]$. $h_{k,i}\in\mathbb{C}$ denotes the channel gain between antenna $i$ and user $k$, which is assumed to remain unchanged in the whole communication process and perfectly known to the server and all users. $\epsilon_k(t)\sim \mathcal{CN}(0,1)$ represents the noise of receiver $k$ at time slot $t$.
	
	A coded caching scenario contains two phases.
	\paragraph*{Placement phase}
	Each user $k\in[K]$ is able to store $MF$ packets from the library, denoted by $\mathcal{Z}_k$, without knowledge of later demands.
	
	\paragraph*{Delivery phase}
	Each user $k\in[K]$ requests an arbitrary file $\mathbf{W}_{d_k}$ where $d_k\in[N]$ from the library. We define $\mathbf{d}\triangleq (d_1,d_2,\ldots,d_K)$ as the demand vector.
	According to the users' demands and caches, the server transmits coded packets through $L$ antennas. More precisely, the server first   uses a code for the Gaussian channel with rate $B/\tilde{B}=\text{log}P+o(\text{log}P)$ (bit per complex symbol), to encode each packet as $\tilde{\mathbf{W}}_{n,f}\triangleq\psi(\mathbf{W}_{n,f})$.	
	By assuming  $P$ is large enough, it can be seen that each coded packet carries one Degree-of-Freedom (DoF). The whole communication process contains $S$ blocks, each of which consists of $\tilde{B}$  complex  symbols (i.e.,  $\tilde{B}$ time slots).
	For each block $s\in[S]$, the server delivers a subset of requested packets, denoted by $\mathcal{D}_s=\{\mathbf{\tilde{W}}_{d_{\mathcal{R}_{s,1}},f_{\mathcal{R}_{s,1}}} \mathbf{\tilde{W}}_{d_{\mathcal{R}_{s,2}},f_{\mathcal{R}_{s,2}}}, \ldots, \mathbf{\tilde{W}}_{d_{\mathcal{R}_{s,r_s}},f_{\mathcal{R}_{s,r_s}}}\}$ to a subset of users denoted by $\mathcal{R}_s=\{\mathcal{R}_{s,1},\mathcal{R}_{s,2},\ldots,\mathcal{R}_{s,r_s}\}$, where $|\mathcal{R}_s|=r_s$.
	In this paper, we only consider linear coding schemes in the delivery phase. Thus
	each antenna $i\in[L]$   sends the following linear combination from $\mathcal{D}_s$ to the users of $\mathcal{R}_s$ in block $s$, given by
	\begin{equation}\label{T-sig}
		\mathbf{x}_i(s)=\sum_{k\in\mathcal{R}_s}v_{i,k}^{(s)}\mathbf{\tilde{W}}_{d_{k},f_k},
	\end{equation}
	where each $v_{i,k}^{(s)}$ is a scalar complex coefficient.
	
	The transmission   by all antennas  in block $s$ could be written as
	{\small
	\begin{eqnarray}
		\label{eq-coding-caching}
		\mathbf{X}(s)=
		\left(
		\begin{array}{c}
			\mathbf{x}_{1}(s)\\
			\mathbf{x}_{2}(s)\\
			\vdots\\
			\mathbf{x}_{L}(s)
		\end{array}
		\right):=\mathbf{V}^{(s)}		\left(
		\begin{array}{c}
			\mathbf{\tilde{W}}_{d_{\mathcal{R}_{s,1}},f_{\mathcal{R}_{s,1}}}\\
			\mathbf{\tilde{W}}_{d_{\mathcal{R}_{s,2}},f_{\mathcal{R}_{s,2}}}\\
			\vdots\\
			\mathbf{\tilde{W}}_{d_{\mathcal{R}_{s,r_s}},f_{\mathcal{R}_{s,r_s}}}
		\end{array}
		\right)
		:=
		\left(
		\begin{array}{cccc}
			v_{1,1}^{(s)}&v_{1,2}^{(s)}&\cdots&v_{1,r_s}^{(s)}\\
			v_{2,1}^{(s)}&v_{2,2}^{(s)}&\cdots&v_{2,r_s}^{(s)}\\
			\vdots&\vdots&\ddots&\vdots\\
			v_{L,1}^{(s)}&v_{L,2}^{(s)}&\cdots &v_{L,r_s}^{(s)}
		\end{array}
		\right)
		\left(
		\begin{array}{c}
			\mathbf{\tilde{W}}_{d_{\mathcal{R}_{s,1}},f_{\mathcal{R}_{s,1}}}\\
			\mathbf{\tilde{W}}_{d_{\mathcal{R}_{s,2}},f_{\mathcal{R}_{s,2}}}\\
			\vdots\\
			\mathbf{\tilde{W}}_{d_{\mathcal{R}_{s,r_s}},f_{\mathcal{R}_{s,r_s}}}
		\end{array}
		\right).
	\end{eqnarray}}
	The signal received by each user  $k\in \mathcal{R}_s$ at block $s$ is denoted by
	\begin{equation}\label{R-sig}
		\mathbf{y}_k(s)=\sum_{i=1}^{L}h_{k,i}^{(s)}\mathbf{x}_i(s)+\mbox{\boldmath$\epsilon$}_k(s) \in \mathbb{C}^{\tilde{B}},
	\end{equation}
	where $\mbox{\boldmath$\epsilon$}_k(s)$ denotes the random noise vector at user $k$ in block $s$.
	Then from the \eqref{T-sig} and \eqref{R-sig}, the received signals by all users in block $s$ can be written as
	{\small
	\begin{align*}
		\mathbf{Y}(s)&=\left(
		\begin{array}{c}
			\mathbf {y}_{\mathcal{R}_{s,1}}(s)\\
			\mathbf {y}_{\mathcal{R}_{s,2}}(s)\\
			\vdots\\
			\mathbf {y}_{\mathcal{R}_{s,r_s}}(s)
		\end{array}			\right):=\mathbf{H}^{(s)}\left(
		\begin{array}{c}
			\mathbf{x}_1(s)\\
			\mathbf{x}_2(s)\\
			\vdots\\
			\mathbf{x}_L(s)
		\end{array}
		\right)+\left(
		\begin{array}{c}
			\mbox{\boldmath$\epsilon$}_1(s)\\
			\mbox{\boldmath$\epsilon$}_2(s)\\
			\vdots\\
			\mbox{\boldmath$\epsilon$}_{r_s}(s)\\
		\end{array}
		\right)\\
		&:=
		\left(
		\begin{array}{cccc}
			h_{1,1}^{(s)}&h_{1,2}^{(s)}&\cdots&h_{1,L}^{(s)}\\
			h_{2,1}^{(s)}&h_{2,2}^{(s)}&\cdots&h_{2,L}^{(s)}\\
			\vdots&\vdots&\ddots&\vdots\\
			h_{r_s,1}^{(s)}&h_{r_s,1}^{(s)}&\cdots &h_{r_s,L}^{(s)}
		\end{array}
		\right)
		\left(
		\begin{array}{cccc}
			v_{1,1}^{(s)}&v_{1,2}^{(s)}&\cdots&v_{1,r_s}^{(s)}\\
			v_{2,1}^{(s)}&v_{2,2}^{(s)}&\cdots&v_{2,r_s}^{(s)}\\
			\vdots&\vdots&\ddots&\vdots\\
			v_{L,1}^{(s)}&v_{L,2}^{(s)}&\cdots &v_{L,r_s}^{(s)}
		\end{array}
		\right)
		\left(
		\begin{array}{c}
			\mathbf{\tilde{W}}_{d_{\mathcal{R}_{s,1}},f_{\mathcal{R}_{s,1}}}\\
			\mathbf{\tilde{W}}_{d_{\mathcal{R}_{s,2}},f_{\mathcal{R}_{s,2}}}\\
			\vdots\\
			\mathbf{\tilde{W}}_{d_{\mathcal{R}_{s,r_s}},f_{\mathcal{R}_{s,r_s}}}
		\end{array}
		\right)
		+
		\left(
		\begin{array}{c}
			\mbox{\boldmath$\epsilon$}_1(s)\\
			\mbox{\boldmath$\epsilon$}_2(s)\\
			\vdots\\
			\mbox{\boldmath$\epsilon$}_{r_s}(s)\\
		\end{array}
		\right)\\
		&:=\left(
		\begin{array}{cccc}
			a_{1,1}^{(s)}&a_{1,2}^{(s)}&\cdots&a_{1,r_s}^{(s)}\\
			a_{2,1}^{(s)}&a_{2,2}^{(s)}&\cdots&a_{2,r_s}^{(s)}	\\
			\vdots &\vdots&\ddots&\vdots\\
			a_{r_s,1}^{(s)}&a_{r_s,2}^{(s)}&\cdots&a_{r_s,r_s}^{(s)}	\\
		\end{array}
		\right)
		\left(
		\begin{array}{c}
			\mathbf{\tilde{W}}_{d_{\mathcal{R}_{s,1}},f_{\mathcal{R}_{s,1}}}\\
			\mathbf{\tilde{W}}_{d_{\mathcal{R}_{s,2}},f_{\mathcal{R}_{s,2}}}\\
			\vdots\\
			\mathbf{\tilde{W}}_{d_{\mathcal{R}_{s,r_s}},f_{\mathcal{R}_{s,r_s}}}
		\end{array}
		\right)
		+
		\left(
		\begin{array}{c}
			\mbox{\boldmath$\epsilon$}_1(s)\\
			\mbox{\boldmath$\epsilon$}_2(s)\\
			\vdots\\
			\mbox{\boldmath$\epsilon$}_{r_s}(s)\\
		\end{array}
		\right)\\
		&:=\mathbf{R}^{(s)}
		\left(
		\begin{array}{c}
			\mathbf{\tilde{W}}_{d_{\mathcal{R}_{s,1}},f_{\mathcal{R}_{s,1}}}\\
			\mathbf{\tilde{W}}_{d_{\mathcal{R}_{s,2}},f_{\mathcal{R}_{s,2}}}\\
			\vdots\\
			\mathbf{\tilde{W}}_{d_{\mathcal{R}_{s,r_s}},f_{\mathcal{R}_{s,r_s}}}
		\end{array}
		\right)
		+
		\left(
		\begin{array}{c}
			\mbox{\boldmath$\epsilon$}_1(s)\\
			\mbox{\boldmath$\epsilon$}_2(s)\\
			\vdots\\
			\mbox{\boldmath$\epsilon$}_{r_s}(s)\\
		\end{array}
		\right),	
	\end{align*}}where $\mathbf{H}^{(s)}$ is random interference channel matrix with dimension $r_s\times L$ in block $s$, and any submatrix of $\mathbf{H}^{(s)}$ with dimension $L\times L$ is invertible with high probability. If each user $k\in\mathcal{R}_s$ can utilize the cache content to subtract the interference from its received signal,  in correspondence of $\mathbf{y}_k(s)$, it ``sees" the output of an equivalent point-to-point Gaussian channel given by ($\tilde{\mathcal{Z}_k}$ denotes the coded packets cached by user $k$)

	\begin{equation}\label{R-req}
		\mathcal{L}_{s,k}(\mathbf{y}_k(s),\mathcal{\tilde{Z}}_k)=\mathbf{\tilde{W}}_{d_k,f_k}+\mbox{\boldmath$\epsilon$}_k(s),
	\end{equation}
	for which the rate that scales as $\text{log}P+o(\text{log}P)$ for large $P$ is achievable. Since $\tilde{\mathbf{W}}_{d_k,f_k}$ is encoded by a rate of $\text{log}P+o(\text{log}P)$, each packet can be decoded with vanishing error probability as $B$ increases. In order to make readers easily understand, we omit the encoding function $\psi$ when we focus
	on introducing our schemes.	

	The one-shot linear  sum-DoF in block $s$ is $r_s$, i.e., the received rate (the sum of all served users in block $s$). Therefore the sum-DoF of the whole system for the demand vector $\mathbf{d}$ is $\frac{\sum_{s=1}^S r_s}{S}$.  A sum-DoF $d(L,N,M,K)$  is said achievable if there exists a two-phase coded caching scheme  with delivery rate $\text{log}P+o(\text{log}P)$ in each time slot,
 	where the sum-DoF of the whole system for each possible demand vector is at least $d(L,N,M,K)$.
	Our objective is to find the maximum (or supremum) of all achievable sum-DoFs.
	While achieving the maximum sum-DoF, the subpacketization of the proposed scheme should be as low as possible.
	
	\section{Multiple-antenna placement delivery array}\label{sec-MAPDA}
	\subsection{Placement delivery array and Latin square}	
	\begin{definition}(\cite{YCTC})\label{def-PDA}
		For  positive integers $K$, $F$, $Z$ and $S$, an $F\times K$ array  $\mathbf{Q}=(\mathbf{Q}(f,k))_{f\in[F],k\in[K]}$, composed of a specific symbol $``*"$  and $S$ positive integers
		$1,2,\cdots, S$, is called a $(K,F,Z,S)$ placement delivery array (PDA) if it satisfies the following conditions:
		\begin{enumerate}
			\item [C$1$.] The symbol $``*"$ appears $Z$ times in each column;
			\item [C$2$.] Each integer occurs at least once in the array;
			\item [C$3$.] For any two distinct entries $\mathbf{Q}(f_1,k_1)=\mathbf{Q}(f_2,k_2)=s$ is an integer  only if
			\begin{enumerate}
				\item [a.] $f_1\ne f_2$, $k_1\ne k_2$, i.e., they lie in distinct rows and distinct columns; and
				\item [b.] $\mathbf{Q}(f_1,k_2)=\mathbf{Q}(f_2,k_1)=*$, i.e., the corresponding $2\times 2$  subarray formed by rows $f_1$, $f_2$ and columns $k_1$, $k_2$ must be of the following form
				\begin{eqnarray*}
					\left(\begin{array}{cc}
						s & *\\
						* & s
					\end{array}\right)~\textrm{or}~
					\left(\begin{array}{cc}
						* & s\\
						s & *
					\end{array}\right).
				\end{eqnarray*}
			\end{enumerate}
		\end{enumerate}
		\hfill $\square$
	\end{definition}
	 Based on a $(K,F,Z,S)$ PDA, an $F$-division coded caching scheme for the $(K,M,N)$ caching system, where $M/N=Z/F$, can be obtained by using Algorithm \ref{alg-PDA}.
		\begin{lemma}(\cite{YCTC})\label{le-PDA}
			Given any   $(K,F,Z,S)$ PDA, there exists an $F$-division caching scheme for the $(K,M,N)$ caching system with memory ratio $\frac{M}{N}=\frac{Z}{F}$, load $  \frac{S}{F}$, and subpacketization $ F$.
			
			\hfill $\square$
		\end{lemma}
		
		\begin{algorithm}
			\caption{Caching scheme based on PDA in \cite{YCTC}}\label{alg-PDA}
			\begin{algorithmic}[1]
				\Procedure {Placement}{$\mathbf{Q}$, $\mathcal{W}$}
				\State Split each file $W_n\in\mathcal{W}$ into $F$ packets, i.e., $W_{n}=\{W_{n,f}\ |\ f=1,2,\cdots,F\}$.
				\For{$k\in [K]$}
				\State $\mathcal{Z}_k\leftarrow\{W_{n,f}\ |\ \mathbf{P}(f,k)=*, n=[N],f=[F]\}$
				\EndFor
				\EndProcedure
				\Procedure{Delivery}{$\mathbf{Q}, \mathcal{W},{\bf d}$}
				\For{$s=1,2,\cdots,S$}
				\State  Server sends $\bigoplus_{\mathbf{Q}(f,k)=s,f\in[F],k\in[K]}W_{d_{k},f}$.
				\EndFor
				\EndProcedure
			\end{algorithmic}
		\end{algorithm}

	\begin{remark}
		From Algorithm \ref{alg-PDA}, the relationships between a $(K,F,Z,S)$ PDA $\mathbf{Q}$ and its realizing coded caching scheme as follows.
		\begin{itemize}
			\item The $K$ columns and $F$ rows denote the users and packets of each file, respectively. The entry $\mathbf{Q}(f,k)=*$ represents that the $f^{\text{th}}$ packet of all files is cached by user $k$. Each user caches $M=\frac{ZN}{F}$ files by Condition C$1$ of Definition \ref{def-PDA}.
			\item The server will broadcast the multicast messages to users in block $s$, i.e., the XOR of all the requested packets which are indicated by $s$ are sent to users. Each user can obtain the demanded packet in block $s$ from Condition C$3$ of Definition \ref{def-PDA}.
			\item Condition C$2$ of Definition \ref{def-PDA} implies that the number of each integer $s$ appears $r_s$ times, i.e., the coded caching gain is $r_s$ in block $s$.
		\end{itemize}
		\hfill $\square$
	\end{remark}
	Specially, if each integer appears $g$ times in the array, the PDA is a  g-regular PDA, denoted by $g$-$(K,F,Z,S)$ PDA.
	The following lemma shows that the MN scheme corresponds to a specific PDA, referred to as MN PDA.

	\begin{lemma}(\cite{MN} MN PDA)
		\label{le-MN}
		For any positive integers $K$ and $t$ with $t<K$, there exists a $(t+1)$-$\left(K,{K\choose t},{K-1\choose t-1},{K\choose t+1}\right)$ PDA.
		\hfill $\square$
	\end{lemma}
		\begin{definition}(\cite{CJ})
			\label{def-Latin}
			A Latin square is an $n \times n $ squared array in which there are exactly $n$ different elements, each of which appears exactly once in each row and column, where $n$ is a positive integer.
			\hfill $\square$
		\end{definition}
It is well known that for any positive integer $n$ there always exists a Latin square of order $n$. For example, when $n=5$ the following square $\mathbf{L}$ is a Latin square.
		\begin{eqnarray}\label{LT-square}
		\mathbf{L}=\left(
			\begin{array}{ccccc}
				1&2&3&4&5\\
				2&3&4&5&1\\
				3&4&5&1&2\\
				4&5&1&2&3\\
				5&1&2&3&4
			\end{array}
		\right).
		\end{eqnarray}
	It can be seen that each integer $s\in[5]$ occurs in each row and column exactly once.
	\subsection{Multiple-antenna Placement Delivery Array}
	In this section, we propose a novel placement delivery array to characterize the multiple antennas coded caching scheme, referred to as multiple-antenna placement delivery array (MAPDA), which combines the concept of PDA with zero-forcing.
	\begin{definition}\rm
		\label{def-MAPDA}
		For any positive integers $L$, $K$, $F$, $Z$ and $S$, an $F\times K$ array $\mathbf{P}$ composed of $``*"$ and $[S]$ is called $(L,K,F,Z,S)$ multiple-antenna placement delivery array (MAPDA) if it satisfies Conditions C$1$, C$2$ in Definition \ref{def-PDA} and
		\begin{itemize}
			\item[C$3$.] Each integer $s$ appears at most once in each column;
			\item[C$4$.] For any integer $s\in[S]$, define  $\mathbf{P}^{(s)} $
			to be the subarray of $\mathbf{P}$ including the rows and columns containing $s$, and let $r'_s\times r_s$ denote the dimensions of $\mathbf{P}^{(s)}$.  The number of integer entries in each row  of $\mathbf{P}^{(s)}$ is less than or equal to $L$, i.e.,
			\begin{eqnarray}\label{C4}
				\left|\{k_1\in [r_s]  |\ \mathbf{P}^{(s)}(f_1,k_1)\in[S]\}\right|\leq L, \ \forall f_1 \in [r'_s].
			\end{eqnarray}
		\end{itemize}
		\hfill $\square$
	\end{definition}
	
	If each integer appears $g$ times in the $\mathbf{P}$, then $\mathbf{P}$ is a g-regular MAPDA, denoted by  $g$-$(L,K,F,Z,S)$ MAPDA.

	\begin{example}\label{ex-1}
		The following array $\mathbf{P}$ is a $4$-$(3,4,4,1,3)$ MAPDA, 	
		\begin{eqnarray}
			\label{ex-MAPDA}
			\mathbf{P}=\left(\begin{array}{cccc}
				* & 1 & 2 & 3 \\
				1 & * & 3 & 2 \\
				2 & 3 & * & 1 \\
				3 & 2 & 1 & *
			\end{array}\right).
		\end{eqnarray}
	Notice that each integer appears once in each column, so that C$3$ is satisfied.	To see that also C$4$ holds, consider for example $s = 1$.
	It can be seen that $\mathbf{P}^{(1)} = \mathbf{P}$ and that each row of $\mathbf{P}^{(1)}$ contains $L = 3$ integer entries  and one star. The same happens for $s = 2$, $3$. Hence also C$4$ is satisfied. 		
		\hfill $\square$
	\end{example}

	Similar to the coded caching scheme realized by a PDA, we can use an MAPDA to generate a multiple antennas coded caching scheme. Specifically given a $(L,K,F,Z,S)$ MAPDA $\mathbf{P}$, we obtain  an $F$-division $(L,K,M,N)$ multiple antennas coded caching scheme with memory size $M=\frac{ZN}{F}$ as follows.
	\begin{itemize}
		\item {\bf Placement phase:} Employing   the placement strategy in Algorithm \ref{alg-PDA}, each file $\mathbf{W}_n$ is divided into $F$ packets with equal size, i.e., $\mathbf{W}_{n}=(\mathbf{W}_{n,f}\ |\ f\in [F])$, and each user $k$ caches the following packets by Line 4 of Algorithm \ref{alg-PDA}.
		\begin{eqnarray}
			\label{eq-caching-content}
			\mathcal{Z}_k=\{\mathbf{W}_{n,f}\ |\ \mathbf{P}(f,k)=*,\ f\in [F],\ n\in [N]\}.
		\end{eqnarray}
		Then each user caches $M=\frac{ZN}{F}$ files;
	
		\item {\bf Delivery phase:} For any request vector ${\bf d}$, similar to the delivery strategy in Algorithm \ref{alg-PDA}, each integer $s\in [S]$ also indicates the multicast messages  sent by the server with $L$ antennas in block $s$ through the MISO broadcast channel according to the channel matrix $\mathbf{H}^{(s)}$ with dimension $r_s\times L$. Assume that there are $r_s$  entries $\mathbf{P}(f_{\mathcal{R}_{s,1}},\mathcal{R}_{s,1})$, $\mathbf{P}(f_{\mathcal{R}_{s,2}},\mathcal{R}_{s,1})$, $\ldots$ , $\mathbf{P}(f_{\mathcal{R}_{s,r_s}},\mathcal{R}_{s,1})$ equal to $s$, where $f_{\mathcal{R}_{s,i}}\in[F]$ and $\mathcal{R}_{s,i}\in[K]$  for each $i\in [r_s]$.
		From Condition C3 in Definition \ref{def-MAPDA}, we have that the column indices $\mathcal{R}_{s,i}$ are distinct, and we can assume without off of generality that
		$\mathcal{R}_{s,1}<\mathcal{R}_{s,2}<\cdots<\mathcal{R}_{s,r_s}$. From \eqref{eq-caching-content} each user $\mathcal{R}_{s,i}$ where $i\in[r_s]$ does not cache its requiring packet $W_{d_{\mathcal{R}_{s,i}},f_{\mathcal{R}_{s,i}}}$, since $\mathbf{P}(f_{\mathcal{R}_{s,i}},\mathcal{R}_{s,i})\neq *$. The vector of  packets to be transmitted in block $s$ and the user set to recover these packets are denoted by
		\begin{equation}
			\label{eq-packet-user-time-s}
			\mathbf{W}^{(s)}=\left(
			\begin{array}{c}
				\mathbf{W}_{d_{\mathcal{R}_{s,1}},f_{\mathcal{R}_{s,1}}}\\
				\mathbf{W}_{d_{\mathcal{R}_{s,2}},f_{\mathcal{R}_{s,2}}}\\
				\vdots\\
				\mathbf{W}_{d_{\mathcal{R}_{s,r_s}},f_{\mathcal{R}_{s,r_s}}}
			\end{array}
			\right),\ \ \ \ \ \ \ \mathcal{R}_s=\{\mathcal{R}_{s,1},\mathcal{R}_{s,2},\ldots,\mathcal{R}_{s,r_s}\},
		\end{equation}
		respectively.
		For each $\mathcal{R}_{s,i}\in \mathcal{R}_s$, assume that there are $l^{(s)}_{i}$ columns with indices in  $\mathcal{R}_s$, which contains integers at the row $f_{\mathcal{R}_{s,i}}$ of $\mathbf{P}$. The column index set is denoted by
		\begin{eqnarray}
			\label{eq-not-cached}
			\mathcal{P}^{(s)}_i=\{\mathcal{R}_{s,i'}\in \mathcal{R}_s\ |\ \mathbf{P}(f_{\mathcal{R}_{s,i}},\mathcal{R}_{s,i'})\in [S],\ i'\in [r_s]\}.
		\end{eqnarray}
		By \eqref{eq-caching-content} and \eqref{eq-not-cached}, the demanded packet $\mathbf{W}_{d_{\mathcal{R}_{s,i}},f_{\mathcal{R}_{s,i}}}$ required by user $\mathcal{R}_{s,i}$ is not cached by user $\mathcal{R}_{s,i'}$   iff $\mathcal{R}_{s,i'}\in\mathcal{P}^{(s)}_i$. In the following we will take the column indices in $\mathcal{R}_s$ and row indices $f_{\mathcal{R}_{s,1}}$, $f_{\mathcal{R}_{s,2}}$, $\ldots$, $f_{\mathcal{R}_{s,r_s}}$ as the columns indices and row indices of the subarray $\mathbf{P}^{(s)}$. Then the integer $l^{(s)}_{i}$ is exactly the number of integer entries at row $f_{\mathcal{R}_{s,i}}$ of $\mathbf{P}^{(s)}$. From   Condition C$4$ of Definition \ref{def-MAPDA}, we have $l^{(s)}_{i}\leq L$. Recall that any $l^{(s)}_{i}$ rows of $\mathbf{H}^{(s)}$ are linear independent with high probability. We take the column indices in $\mathcal{R}_s$ and row indices $f_{\mathcal{R}_{s,1}}$, $f_{\mathcal{R}_{s,2}}$, $\ldots$, $f_{\mathcal{R}_{s,r_s}}$ as the  column   and row  indices of $\mathbf{H}^{(s)}$, respectively.  From linear algebra, for each $i\in[r_s]$ we can get a column vector ${\bf v}^{(s)}_i$ such that
		\begin{eqnarray}
			\label{eq-coding-vector}
			\mathbf{H}^{(s)}(\mathcal{R}_{s,i}) {\bf v}^{(s)}_i=1,\ \ \ \mathbf{H}^{(s)}(\mathcal{R}_{s,i'}) {\bf v}^{(s)}_i=0,\ \ \ \forall i'\in\mathcal{P}^{(s)}_i\setminus\{i\},
		\end{eqnarray} where $\mathbf{H}^{(s)}(\mathcal{R}_{s,i})$ and $\mathbf{H}^{(s)}(\mathcal{R}_{s,i'})$ are the rows with indices $\mathcal{R}_{s,i}$  and $ \mathcal{R}_{s,i'}$ of $\mathbf{H}^{(s)}$, respectively.  We define that the precoding matrix is
		$
		\mathbf{V}^{(s)}=\left({\bf v}^{(s)}_1,{\bf v}^{(s)}_2,\ldots,{\bf v}^{(s)}_{r_s} \right).
		$
		Thus the transmitted messages in block $s$ by the server are
		\begin{eqnarray*}
			\mathbf{X}(s)=\mathbf{V}^{(s)}\mathbf{W}^{(s)}.
		\end{eqnarray*}
		The received messages by the users in $\mathcal{R}_s$ over the multiple antennas broadcast channel in block $s$ are
		\begin{eqnarray*}
			\mathbf{Y}_s&=&\mathbf{H}^{(s)}\mathbf{X}(s)=\mathbf{H}^{(s)}\mathbf{V}^{(s)}\mathbf{W}^{(s)}\\
			&=&\mathbf{H}^{(s)}\left({\bf v}^{(s)}_1,{\bf v}^{(s)}_2,\ldots,{\bf v}^{(s)}_{r_s} \right)
			\mathbf{W}^{(s)}\\
			&=&\left(\mathbf{H}^{(s)}{\bf v}^{(s)}_1,\mathbf{H}^{(s)}{\bf v}^{(s)}_2,\ldots,\mathbf{H}^{(s)}{\bf v}^{(s)}_{r_s} \right)
			\mathbf{W}^{(s)}\\
			&:=&\mathbf{R}^{(s)}\mathbf{W}^{(s)}.
		\end{eqnarray*}
		By \eqref{eq-coding-vector},  each column of $\mathbf{R}^{(s)}$ has at least $l^{(s)}_{i}-1$ zero entries. Furthermore, we have $\mathbf{R}^{(s)}(\mathcal{R}_{s,i},i)=1$ and $\mathbf{R}^{(s)}(\mathcal{R}_{s,i},i')=0$ for each $i'\in \mathcal{P}^{(s)}_i\setminus\{i\}$. This implies that for any required packet $\mathbf{W}_{d_{\mathcal{R}_{s,i'}},f_{\mathcal{R}_{s,i'}}}$is not cached by user $\mathcal{R}_{s,i}$,  we have $\mathbf{R}^{(s)}(\mathcal{R}_{s,i},i')=0$ if $i'\neq i$ and $\mathbf{R}^{(s)}(\mathcal{R}_{s,i},i')=1$ if $i=i'$.
		Thus in the received  message of user $\mathcal{R}_{s,i}$, there only exist its required packet and cached packets (i.e., the interference packets are zero forced), and thus it can decode the required packet.
		
		Since the number of uncached packets of each file by each user is the same which is equal to $F-Z$, we have $\sum_{s=1}^S r_s=K(F-Z)$. Hence,
		the sum-DoF in the whole procedure is $\frac{\sum_{s=1}^S r_s}{S}=\frac{K(F-Z)}{S}$.
	\end{itemize}
		
	We then continue Example~\ref{ex-1} to illustrate the resulting   scheme of the MAPDA in~\eqref{ex-MAPDA}.
	\begin{example}	\label{ex-2}
	Using the $(3,4,4,1,3)$ MAPDA $\mathbf{P}$ in \eqref{ex-MAPDA}, we can obtain a multiple antennas coded caching scheme as follows.
	\begin{itemize}
		\item {\textbf{Placement phase:}} Divide each file into $4$ packets with equal size, i.e., $\mathbf{W}_n=\{\mathbf{W}_{n,1}$, $\mathbf{W}_{n,2}$, $\mathbf{W}_{n,3}$, $\mathbf{W}_{n,4})$, $n\in[4]$. By \eqref{eq-caching-content}, the caches of users are
		\begin{eqnarray*}
			&&\mathcal{Z}_1=\{\mathbf{W}_{n,1}\ |\ n\in[4]\},\ \ \ \ \mathcal{Z}_2=\{\mathbf{W}_{n,2}\ |\ n\in[4]\},\\
			&&\mathcal{Z}_3=\{\mathbf{W}_{n,3}\ |\ n\in[4]\},\ \ \ \ \mathcal{Z}_4=\{\mathbf{W}_{n,4}\ |\ n\in[4]\}.
		\end{eqnarray*}
				
		\item {\textbf {Delivery phase:}} Assume that the request vector is $\textbf{d}=(1,2,3,4)$. From \eqref{eq-packet-user-time-s}, all the required packet to be transmitted in the $3$ blocks are				
		\begin{equation*}
			\mathbf{W}^{(1)}=\left(
			\begin{array}{c}
				\mathbf{W}_{1,2}\\
				\mathbf{W}_{2,1}\\
				\mathbf{W}_{3,4}\\
				\mathbf{W}_{4,3}
			\end{array}
			\right),\ \
			\mathbf{W}^{(2)}=\left(
			\begin{array}{c}
				\mathbf{W}_{1,3}\\
				\mathbf{W}_{2,4}\\
				\mathbf{W}_{3,1}\\
				\mathbf{W}_{4,2}
			\end{array}
			\right),\ \ \mathbf{W}^{(3)}=\left(
			\begin{array}{c}
				\mathbf{W}_{1,4}\\
				\mathbf{W}_{2,3}\\
				\mathbf{W}_{3,2}\\
				\mathbf{W}_{4,1}
			\end{array}
			\right),
		\end{equation*}
		while the sets of users who are involved in the $3$ blocks are
		$
			\mathcal{R}_{1 }=\mathcal{R}_{2 }=\mathcal{R}_{3 }=\{1,2,3,4\},
		$
		respectively. Next we consider block $1$, where we have $r_1=4$ and
		\begin{eqnarray*}
			\mathcal{P}^{(1)}_1=\{1,3,4\},\ \ \mathcal{P}^{(1)}_2=\{2,3,4\},\ \ \mathcal{P}^{(1)}_3=\{1,2,3\},\ \ \mathcal{P}^{(1)}_4=\{1,2,4\}.
		\end{eqnarray*}
		Let the precoding matrix be
		$\mathbf{V}^{(1)}=({\bf v}^{(1)}_1, {\bf v}^{(1)}_2,{\bf v}^{(1)}_3,{\bf v}^{(1)}_4)$, where each ${\bf v}^{(1)}_i$ satisfies \eqref{eq-coding-vector} for $i\in [4]$, i.e., the server sends the messages $\mathbf{X}(1)=\mathbf{V}^{(1)}\mathbf{W}^{(1)}$ and the users in $\mathcal{R}_{1 }$ receive
		\begin{eqnarray*}
			\mathbf{Y}_1&=&\mathbf{H}^{(1)}\mathbf{X}(1)=\mathbf{H}^{(1)}\mathbf{V}^{(1)}\mathbf{W}^{(1)}\\
			&=&\left(
			\begin{array}{cccc}
				1        & a_{1,2}^{(1)} & 0        & 0 \\
				a_{2,1}^{(1)} & 1        & 0        & 0 \\
				0 & 0        & 1        & a_{3,4}^{(1)}\\
				0 & 0        & a_{4,3}^{(1)} & 1 \\
			\end{array}
			\right)\left(
			\begin{array}{c}
				\mathbf{W}_{1,2}\\
				\mathbf{W}_{2,1}\\
				\mathbf{W}_{3,4}\\
				\mathbf{W}_{4,3}
			\end{array}
			\right)
			= \left(
			\begin{array}{c}
				\mathbf{W}_{1,2}+a_{1,2}^{(1)}\mathbf{W}_{2,1}\\
				a_{2,1}^{(1)}\mathbf{W}_{1,2}+\mathbf{W}_{2,1}\\
				\mathbf{W}_{3,4}+a_{3,4}^{(1)}\mathbf{W}_{4,3}\\
				a_{4,3}^{(1)}\mathbf{W}_{3,4}+\mathbf{W}_{4,3}
			\end{array}
			\right).
		\end{eqnarray*}
		In the above transmission, all the  users can get their required packets by  using their caches. For example, user $1$ can obtain $\mathbf{W}_{1,2}$ from the first component of $\mathbf{Y}_1$ by subtracting $\mathbf{W}_{2,1}$, which is contained in user $1$ cache.
		\end{itemize}	
		Since  $4$ packets are transmitted by the server to satisfy $4$ users in each block,   the sum-DoF is $4=1+3=\frac{KM}{N}+L$.
			
		\hfill $\square$
	\end{example}

	From the above analysis, we can obtain the following result.
	\begin{theorem}\label{th-MAPDA-CCS}
		For a given $(L,K,F,Z,S)$ MAPDA $\mathbf{P} $,  there exists an  $F$-division scheme for the   $(L,K,M,N)$ multiple antennas coded caching  problem with memory ratio $\frac{M}{N}=\frac{Z}{F}$, sum-DoF $\frac{K(F-Z)}{S}$ and subpacketization $F$.
		\hfill $\square$
	\end{theorem}
	In the literature, the schemes in~\cite{NMA,EP,SPSET,MB,STSK} can be represented by MAPDAs.
	
	Under the constraints of uncoded cache placement and one-shot linear delivery, the maximum sum-DoF is upper bounded than $\frac{KM}{N}+L$~\cite{EBPresolving}, which is also an upper bound  on the sum-DoF achieved by the caching schemes from MAPDA. 	
	\begin{theorem}[\cite{EBPresolving}]\label{th-DoF}
		The sum-DoF of an $(L,K,M,N)$ multiple antennas coded caching scheme with memory ratio $\frac{M}{N}=\frac{Z}{F}$ realized by  any $(L,K,F,Z,S)$ MAPDA  is no more than $\frac{KZ}{F}+L=\frac{KM}{N}+L$.
		\hfill $\square$
	\end{theorem}
In Appendix \ref{proof-DoF}, we provide an alternative proof for Theorem~\ref{th-DoF}, based on the	combinatorial structure of MAPDA.
	
	\section{Two novel constructions of MAPDA}\label{sec-const}
	\subsection{Main results and performance analysis}
 	In this  section, we propose two classes of MAPDAs which  lead to   coded caching schemes with the maximum sum-DoF in Theorem~\ref{th-DoF} and significantly reduce subpacketization compared to the state-of-the-art schemes.
		
	We first consider the case where $\frac{KM}{N}+L=K$, and propose the following MAPDA.
	\begin{theorem}\label{th-LT}
		For any positive integers $L$, $K$, $M$ and $N$ with $\frac{KM}{N}+L=K$, there exists a $(\frac{KM}{N}+L)$-$(L,K,K,\frac{KM}{N},K(1-\frac{M}{N}))$ MAPDA which leads to an $(L,K,M,N)$ multiple antennas coded caching scheme with sum-DoF   $\frac{KM}{N} +L=K$ and subpacketization $F=K$.
		\hfill $\square$
	\end{theorem}
	
	\begin{proof}
		Consider a $K\times K$ Latin square $\mathbf{L}$, with elements in $[K]$. Replace the integers $[L+1:K]$ in $\mathbf{L}$ by stars. The resulting array $\mathbf{P}$ is an MAPDA with parameters $(\frac{KM}{N}+L)$-$(L,K,F=K,Z=\frac{KM}{N},S=K(1-\frac{M}{N}))$. This is proved by showing that Conditions C$1$-C$4$ of Definition \ref{def-MAPDA} are satisfied with these parameters. In particular,
		there are $S=K-\frac{KM}{N}$ different integers, each of which appears exactly $K$ times. In addition,  $Z=\frac{KM}{N}$ stars appear in each row and each column. Thus $\mathbf{P}$ satisfies Conditions C$1$ and C$2$ of Definition \ref{def-MAPDA}.
		Since each integer appears once in each column of Latin square, $\mathbf{P}$ satisfies Condition C$3$ of Definition \ref{def-MAPDA}.
		For each integer $s\in [S]$, we have $\mathbf{P}^{(s)}=\mathbf{P}$ always holds.
		So by the structure of Latin square, Condition C$4$ of Definition \ref{def-MAPDA} is also satisfied.
		Thus $\mathbf{P}$ is a $(\frac{KM}{N}+L)$-$(K,K,L,\frac{KM}{N},K-\frac{KM}{N})$ MAPDA.
	\end{proof}	
	For instance, we can obtain the following $5$-$(2,5,5,3,3)$ MAPDA $\mathbf{P}$ by replacing the integers $3$, $4$ and $5$ of Latin square $\mathbf{L}$ in \eqref{LT-square} by stars,  after some rows permutations in order to evidence the cyclic nature of the placement.
		\begin{eqnarray*}
			\mathbf{P}=\left(
			\begin{array}{ccccc}
				1&2&*&*&*\\
				*&1&2&*&*\\
				*&*&1&2&*\\
				*&*&*&1&2\\
				2&*&*&*&1
			\end{array}
			\right)_{5\times 5}.
	\end{eqnarray*}

	Note that our scheme in Theorem \ref{th-LT} achieves the same sum-DoF and subpacketization level as the scheme in \cite{SPSET} but with the constraint
	$\frac{KM}{N}+L=K$, instead of the constraint $L\geq \frac{KM}{N}$ in \cite{SPSET}.

  	We turn now to the general case of $K$, $M$, $N$ and $L$ and given our second main result, which consists of a non-trivial transformation of a g-regular PDA into an MAPDA, as stated in the following theorem, the proof of which is provided in Section \ref{proof-regular}.
	\begin{theorem}\label{th-MAPDA-regular}	
		Given any $g$-$(K_1,F_1,Z_1,S_1)$ PDA, there exists an $(L,mK_1,\alpha F_1,\alpha Z_1,\text{sgn}(g)lS_1)$ MAPDA where $m\leq L$, which leads to an $(L,K=mK_1,M,N)$ multiple antennas coded caching scheme with memory ratio $\frac{M}{N}=\frac{Z_1}{F_1}$, sum-DoF $m(g-1)+L$ and subpacketization $F=\alpha F_1$, where $\alpha=(\text{sgn}(g)+\frac{L-m}{m})l$, $l=\frac{m}{\text{gcd}(m,L-m)}$, and
		\begin{eqnarray}\label{eq-sgn(g)}
			\text{sgn}(g) := \begin{cases}
				1,&\text{if }L=m;\\
				g,&\text{otherwise}.
			\end{cases}
		\end{eqnarray}
		\hfill $\square$
	\end{theorem}
	
	From Theorem \ref{th-MAPDA-regular}, in order to obtain an MAPDA, we only need to construct an appropriate PDA. For instance, starting from the $(t_1+1)$-$(K_1,\binom{K_1}{t_1},\binom{K_1-1}{t_1-1},\binom{K_1}{t_1+1})$ MN PDA and applying the construction of Theorem \ref{th-MAPDA-regular}, we obtain the following result:
	
	\begin{theorem}\label{th-MAPDA-MN}	
		For any positive integers $m$, $L$, $K_1$ and $t_1$  with $t_1<K_1$ and $m\leq L$, there exists an $(L$, $mK_1$, $\alpha {K_1\choose t_1}$, $\alpha {K_1-1\choose t_1-1}$, $\text{sgn}(t_1+1)l{K_1\choose t_1+1})$ MAPDA which leads to an $(L,K=mK_1,M ,N )$ multiple antennas coded caching scheme with the memory ratio $\frac{M}{N}=\frac{t_1}{K_1}$, sum-DoF   $\frac{KM}{N}+L$, and subpacketization $F=\alpha {K_1\choose t_1}$, where $\alpha =(\text{sgn}(t_1+1)+\frac{L-m}{m})l$ and $l=\frac{m}{\text{gcd}(m,L-m)}$.
		
		\hfill $\square$
	\end{theorem}

	\begin{remark}[Two extreme cases in Theorem~\ref{th-MAPDA-MN}]
		\label{re-special-case}
		When $m=L$, we have $\text{sgn}\left(t_1+1\right)=1$, $l=1$, and $\alpha =1$ in Theorem \ref{th-MAPDA-MN}. Then we can obtain a $(L,K,M,N)$ multiple antennas coded caching scheme with the number of users $K=mK_1$, the memory ratio $\frac{M}{N}=\frac{t_1}{K_1}$, subpacketization $F={K_1\choose t_1}$, and sum-DoF $\frac{KM}{N}+L$.
		In other words, under the same constraint as in \cite{EP} (i.e., $L$ divides $K$ and $\frac{KM}{N}$), our scheme is exactly the scheme  in \cite{EP}.

  		When $m=1$, we obtain an  $(L, K ,M ,N )$ multiple antennas coded caching scheme without any constraints, whose subpacketization has the same exponential order as
  		$\binom{K}{KM/N}$  in terms of $K$. Note that  the coded caching scheme in \cite{MB} achieves  the sum-DoF  $\frac{KM}{N}+L$ with subpacketization $\binom{K}{KM/N}$ under the constraint that $\frac{KM}{N}+1$ divides $\frac{KM}{N}+L$.
		\hfill $\square$
	\end{remark}

	Note that 	when  $L\geq \frac{KM}{N}$, the scheme in   \cite{SPSET} achieves the maximum sum-DoF in Theorem~\ref{th-DoF} with linear subpacketization with $K$, which is hard to be further reduced. The main contribution of Theorem~\ref{th-MAPDA-MN} is to reduce the subpacketzation for the case  $L< \frac{KM}{N}$, while achieving the maximum sum-DoF.
	In the following, we compare the performance of Theorem~\ref{th-MAPDA-MN} with the schemes in \cite{NMA,SCH}, since the comparison to   schemes in \cite{EP,MB} has been discussed in Remark~\ref{re-special-case}.
	Furthermore, it can be seen from Table \ref{tab-compare}, the subpacketization of scheme in \cite{SCH} is the less than that of \cite{NMA}, and thus we only need to compare the subpacketization of our scheme in Theorem \ref{th-MAPDA-MN} (denoted by $F_{\text{Th}_5}$) with the scheme  in \cite{SCH} (denoted by $F_{\text{SCH}}$).   The ratio of the subpacketizations $F_{\text{SCH}}$ and $F_{\text{Th}_5}$ is
	\begin{align*}
		\frac{F_{\text{SCH}}}{F_{\text{Th}_5}}=\frac{\binom{K}{t}\binom{K-t-1}{L-1}}{\alpha\binom{K/m}{t/m}}
		&\approx
		\frac{1}{\alpha}\cdot\frac{2^{KH(\frac{t}{K})}2^{(K-t-1)H(\frac{L-1}{K-t-1})}}{2^{\frac{K}{m}H(\frac{t}{K})}}\ \ \ \ (K\rightarrow \infty)\\
		&=\frac{1}{\alpha}\cdot 2^{\frac{(m-1)}{m}KH(\frac{t}{K})+(K-t-1)H(\frac{L-1}{K-t-1})}.
	\end{align*}

	We then provide some numerical evaluations to compare the proposed scheme  in Theorem \ref{th-MAPDA-MN} with the schemes in \cite{SCH,MB}. In Fig.~\ref{fig-performance1}, we consider the  multiple antennas coded caching problem with $K=100$ and $L=7$. In this case, we have $m=5$ for our scheme in Theorem \ref{th-MAPDA-MN}. In Fig.~\ref{fig-performance1}, We do not plot the scheme in \cite{MB}  since it can only work when $\frac{t+L}{t+1} $ is an integer, i.e., only when $t=5$. In this case, the subpacketization of the scheme in \cite{MB} is $7.53\times 10^{7}$, while our scheme in Theorem \ref{th-MAPDA-MN} is $240$.
	In addition, as illustrated in Fig.~\ref{fig-performance1}, compared with the scheme in \cite{SCH} our scheme in Theorem \ref{th-MAPDA-MN} has a significant reduction in subpacketization.
	\begin{figure}
		\centering
		\includegraphics[width=5in]{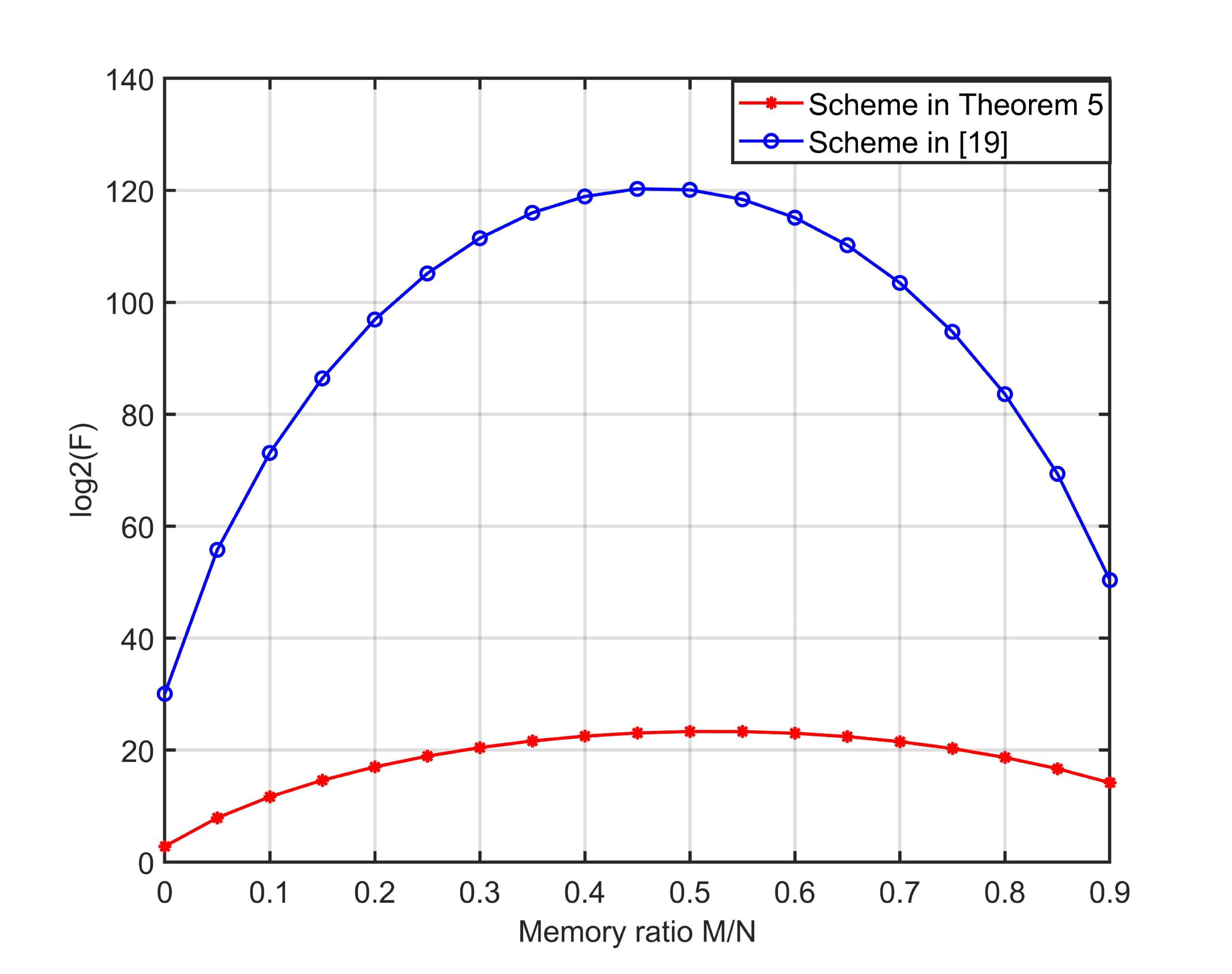}
		\caption{ The subpacketization versus memory ratio for schemes in Theorem \ref{th-MAPDA-MN} and in \cite{SCH}, when $K=100$ and $L=7$.}
		\label{fig-performance1}
	\end{figure}
	
	 \begin{remark}[Extension to multiple cache-aided transmitters wireless channels]
		\label{re-extension}
		The coded caching problem for multi-transmitter wireless interference networks was   considered in \cite{NMA}, where $K_T$ cache-aided single-antenna transmitters (with memory size $M_T$) are connected to $K$ cache-aided users (with memory size $M$) through a wireless interference network.
		It can be seen that when $M_T=N$, it reduces to the cache-aided MISO broadcast   problem in our paper. In addition,
		it was shown in \cite{EP} that when $K_T M_T \geq 1$, any cache-aided MISO scheme could be extended to the   multi-transmitter wireless interference networks by using a cyclic cache placement at the transmitters. By the same approach as in \cite{EP}, we can also extend the proposed schemes in
		Theorems~\ref{th-LT}  and~\ref{th-MAPDA-MN} to achieve the sum-DoF $\frac{K_TM_T}{N}+\frac{KM}{N}$, while the subpacketizations are with a linear order of those in   Theorems~\ref{th-LT}  and~\ref{th-MAPDA-MN}.
		\hfill $\square$
	\end{remark}

 	\begin{remark}
		\label{re-multiserver}
		As already noted in~\cite{SSB}, analogous results in term of the network load (number of equivalent file transmissions to satisfy all users' requests~\cite{MN}) are immediately obtained for a noiseless linear network over a sufficiently large finite field, as induced by end-to-end linear network coding.
		Consider a network of general topology, connecting $ L$ servers and $K$ users. If routing in the intermediate notes is replaced by linear network coding (e.g., each intermediate node forwards random linear combinations of the incoming packets~\cite{1705002}), for sufficiently large  finite field size the resulting $L$-input $K$-output network can be represented as a $K \times L$ matrix of rank $\min\{K,L\}$. Then, the very same ideas can be applied by replacing linear combinations over the complex field with linear combinations over the appropriate finite field. This represents an attractive and general ``separation'' approach between caching and networking, when the routing layer is replaced by linear network coding.
		\hfill $\square$
	\end{remark}
	\subsection{Sketch of the proposed scheme in Theorem \ref{th-MAPDA-MN} }
	\label{sub:sketch}
	Given a $(t_1+1)$-$(K_1,F_1,Z_1,S_1)$ MN PDA $\mathbf{Q}$, for any positive integers $m$ and $L$ with $m\leq L$ and $t_1<K_1$, we will construct an $(L+mt_1)$-$(L,mK_1$, $\alpha F_1$, $\alpha Z_1$, $\text{sgn}(t_1+1)lS_1)$ MAPDA $\mathbf{P}$.
	The main idea of construction is replicating an $K_1$-user MN  PDA   $m$ times, such that each integer appears $m(t_1+1)$ times. If $m(t_1+1)<L+mt_1$, i.e, the achieved sum-DoF is less than  $L+mt_1$ (which is the task sum-DoF we want to achieve), we then  replicate $\mathbf{Q}$ and adjust some integers such that each integer appears $L+mt_1$ times and all Conditions of Definition \ref{def-MAPDA} hold.
	Then we will propose an example to illustrate the detail steps.

	\begin{example}\label{ex-3}
		We focus on the base array $\mathbf{Q}$, which is a $(t_1+1)\text{-}(K_1,F_1,Z_1,S_1)=3$-$(4,6,3,4)$ PDA with $4$ users, memory ratio $\frac{M_1}{N_1}=\frac{Z_1}{F_1}=\frac{1}{2}$ and subpacketization $F_1=6$, where $t_1=2$.
		We will transform $\mathbf{Q}$ to an $(mt_1+L)$-$(L,mK_1,F,Z,S)=7\text{-}(3,8,42,21,24)$ MAPDA $\mathbf{P}$ through the following three steps where $m=2$, as illustrated in Fig.~\ref{fig-ex1}.
	\begin{figure}
		\centering
		\includegraphics[width=7in]{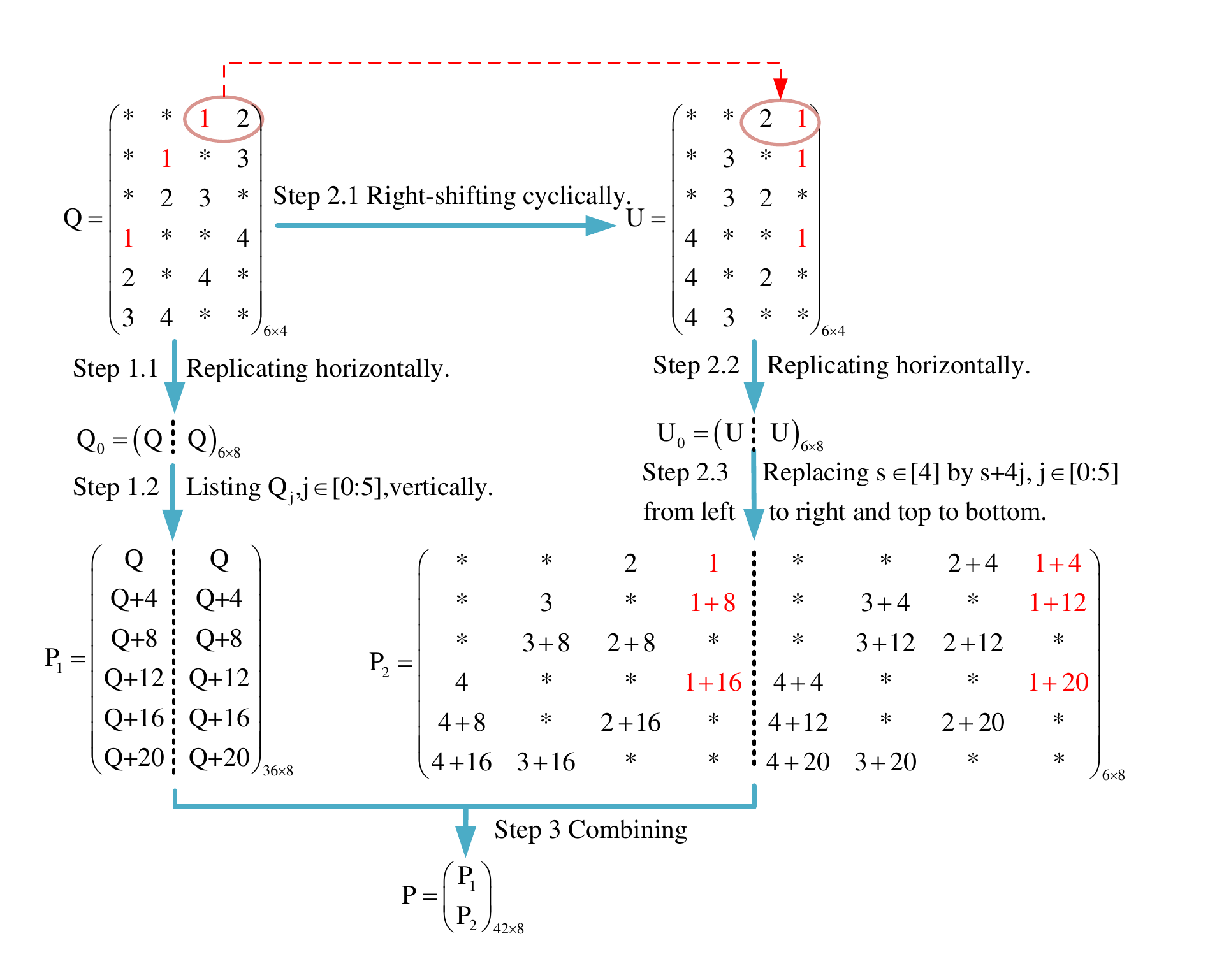}
		\caption{ The transformation from $\mathbf{Q}$ to $\mathbf{P}$ in Theorem \ref{th-MAPDA-MN}.}
		\label{fig-ex1}
	\end{figure}		
	\begin{itemize}
		\item \textbf{Step 1.}  Construction of $\mathbf{P}_1$ from $\mathbf{Q}$.
		\begin{itemize}
			\item \textbf{Step 1.1.}
			We first get a $6\times 8$ array $\mathbf{Q}_0$ by replicating $\mathbf{Q}$   $m=2$ times horizontally.
			\item \textbf{Step 1.2.}
			We then get a $36\times 8$ array $\mathbf{P}_1$ by listing $\mathbf{Q}_0 +4j$ for each $0\leq j<m(t_1+1)=6$. We can see that  there are $S=S_1 m(t_1+1)=24$ different integers in $\mathbf{P}_1$; each integer $s\in[24]$ appears $m(t_1+1)= 6$ times and in distinct columns; each column has $m(t_1+1)Z_1=18$ stars; each row has $mt_1= 4$ stars. Furthermore, the number of integer entries in each row of $\mathbf{P}^{(s)}_1$, where $s\in [24]$, is equal to $ 2<L=3$.
			Thu	all Conditions of Definition \ref{def-MAPDA} are satisfied; $\mathbf{P}_1$ is an MAPDA.
		\end{itemize}
	
		However, we aim to achieve the maximum sum-DoF 	$L+mt_1=7>m(t_1+1)=6$, i.e., each integer should appear $7 $ times in the MAPDA, which motivates the following steps.
		\item \textbf{Step 2.}  Construction of $\mathbf{P}_2$ from $\mathbf{Q}$.
		\begin{itemize}
			\item \textbf{Step 2.1.}
			We replicate $\mathbf{Q}$, $L+mt_1- m(t_1+1) =L-m$ time(s) in this step, such that together with the above step we can achieve the sum-DoF
			$L+mt_1$. In this example, since $L-m=1$, we only need to replicate  $\mathbf{Q}$ once in this step.
			For each row of  $\mathbf{Q}$, we cyclically right-shift  the integer entries by one position, while the stars remain unchanged. The resulting array (denoted by $\mathbf{U}$) has the same dimension as $\mathbf{Q}$. For example  in the first row,    integer $1$ is shifted into the last column, while   integer $2$ is shifted into the $3^{\text{th}}$ column (which is the original column of   integer $1$). It will be clear that this step is used to guarantee Condition C$3$ of Definition \ref{def-MAPDA}.
		
			\item \textbf{Step 2.2.}  $\mathbf{U}_0$ is obtained by replicating   $\mathbf{U}$   $m=2$ times  horizontally. Each integer $s\in[4]$ appears $(t_1+1)=3$ times in $\mathbf{U}$, and thus appears $2(t_1+1)=6$ in     $\mathbf{U}_0$.
			
			\item \textbf{Step 2.3.} By replacing the six integer $s\in[4]$ by $s$, $s+4$, $s+2\times 4$, $s+3\times 4$, $s+4\times 4$, $s+5\times 4$, respectively from left to right and top to bottom, we get a new $6\times 8$ array $\mathbf{P}_2$. For example, the $6$ integer $1$'s in $\mathbf{U}_0$ are replaced by integers $1$, $1+4$, $1+8$, $1+12$, $1+16$, $1+20$, respectively. There are $24$ different integers and each integer appears  once.
			Thus $\mathbf{P}_2$ is also an MAPDA.
		\end{itemize}
		
		\item $\textbf{Step 3.}$  Construction of $\mathbf{P}$ from $\mathbf{P}_1$ and $\mathbf{P}_2$.
		We obtain the $\mathbf{P}$ with $F=7  F_1 =42$ rows and $K= 2  K_1 =8$ columns by concatenating $\mathbf{P}_1$ and $\mathbf{P}_2$ vertically. Next we will show that $\mathbf{P}$ is an MAPDA.
		
		From the above steps, $\mathbf{P}_1$ and $\mathbf{P}_2$ both are MAPDA; thus $\mathbf{P}$ satisfies Conditions C$1$ and C$2$ of Definition \ref{def-PDA}.
		The column indices where each integer $s\in[24]$ lies in of $\mathbf{P}_2$ are different from that of $\mathbf{P}_1$ by the constructions in Steps 2.1 and  2.3; thus $\mathbf{P}$ satisfies Conditions C$3$ of Definition \ref{def-MAPDA}.
		Finally we focus on the $\mathbf{P}^{(s)}$ and verify Condition C$4$ of Definition \ref{def-MAPDA}.
		Let us take  $\mathbf{P}^{(1)}$ as an example,
		$$\mathbf{P}^{(1)}=\left(
		\begin{array}{ccccccc}
			*&	*&	\textcolor{red}{1}&	2&	*&	*&	\textcolor{red}{1}\\
			*&	\textcolor{red}{1}&	*&	3&	*&	\textcolor{red}{1}&	*\\
			\textcolor{red}{1}& *&	*&	4&	\textcolor{red}{1}&	*&	*\\
			*&	*&	2&	\textcolor{red}{1}&	*&	*&	6
		\end{array}
		\right)_{4\times 7}.
		$$
		By the construction, each row of $\mathbf{P}^{(1)}$ has $m$ integer $1$'s and $L-m$ other integer(s), totally $L $ integers. Thus  Condition C$4$ of Definition \ref{def-MAPDA} is satisfied.
	\end{itemize}
	In conclusion,	$\mathbf{P}$ is a $7$-$(3,8,42,21,24)$ MAPDA, which
	leads to a $(L,mK_1,M,N)=(3,8,4,8)$ multiple antennas coded caching scheme with $mK_1=8$ users, memory ratio $\frac{M}{N}=\frac{1}{2}$,  sum-DoF $mt_1+L=7$, and subpacketization $F=42$. In this example, the schemes in ~\cite{SCH,NMA} achieve the sum-DoF $mt_1+L=7$ with subpacketizations $10080$ and $210$ respectively.
	\hfill $\square$ 	
\end{example}

	\section{Proof of Theorem~\ref{th-MAPDA-regular}}
	\label{proof-regular}	
	Given a $g$-$(K_1,F_1,Z_1,S_1)$ PDA $\mathbf{Q}$, we will construct an $(m(g-1)+L)$-$(L,K=mK_1,F,Z,S)$ MAPDA $\mathbf{P}$ for any positive integers $m\leq L$ such that $Z/F=Z_1/F_1$. We first obtain an $F_1\times mK_1$ array $\mathbf{Q}_0$ by replicating $\mathbf{Q}$   $m$ times horizontally, i.e.,
	\begin{eqnarray*}
		\mathbf{Q}_0=\left( \underbrace{\mathbf{Q},\ldots,\mathbf{Q}}_{m}\right).
	\end{eqnarray*}It is straightforward to see that $\mathbf{Q}_0$ satisfies Conditions C$1$, C$2$, C$3$ of Definition \ref{def-MAPDA}. We then focus on Condition C$4$ of Definition \ref{def-MAPDA}.
	By Condition C$3$ of Definition \ref{def-PDA}, in the   subarray of $\mathbf{Q}$ including the rows and columns containing $s$, each row only contains one $s$ while the other elements in this row is $*$. So in the   subarray of $\mathbf{Q}_0$ including the rows and columns containing $s$,  the number of  $s$ in each row is $m$ while the other elements in this row is $*$. Hence, $\mathbf{Q}_0$ satisfies Condition C$4$ of Definition \ref{def-MAPDA}; i.e., $\mathbf{Q}_0$ is an $mg$-$(L,mK_1,F_1,Z_1,S_1)$ MAPDA.
	When $m=L$, the sum-DoF  of $\mathbf{Q}_0$ is $mg=L+m(g-1)$, which is the maximum sum-DoF in Theorem~\ref{th-MAPDA-CCS}.
	In the rest of this section, we only need to consider the case $m<L$.
	
	When $m<L$, we have $mg<L+m(g-1)$, i.e., the achieved sum-DoF $mg$ of $\mathbf{Q}_0$ is less than our task sum-DoF $L+m(g-1)$. In order to increase the sum-DoF by $L+m(g-1)-mg=L-m$,  further steps are taken. In short, we first generate an $mg$-$(L,mK_1,glF_1,glZ_1,glS_1)$ MAPDA $\mathbf{P}_1$ by replicating $\mathbf{Q}_0$   $gl$ times vertically,  where $l=\frac{m}{\text{gcd}(L-m,m)}$,  and adjusting some integers appropriately. Then we generate another  $(L-m)$-$(L$, $mK_1, \frac{(L-m)l}{m}F_1, \frac{(L-m)l}{m}Z_1, glS_1)$ MAPDA $\mathbf{P}_2$ based on $\mathbf{Q}$. The set of integers in $\mathbf{P}_2$  is the same as that in $\mathbf{P}_1$.
	Finally $\mathbf{P}$ is obtained by concatenating $\mathbf{P}_1$ and $\mathbf{P}_2$ vertically.

	In the following,  we introduce the constructions of $\mathbf{P}_1$, $\mathbf{P}_2$ and $\mathbf{P}$ in details.

	\subsection{Construction of an MAPDA $\mathbf{P}_1$}
	\label{sub-P1}
	In order to make the  set of integers in $\mathbf{P}_1$ the same as that in $\mathbf{P}_2$ constructed in the following subsection,
	we construct $\mathbf{P}_1$ by replicating $\mathbf{Q}_0$   $gl$ times vertically and then increasing the integers in $\mathbf{Q}_0$ by the occurrence orders (from up to down) of $\mathbf{Q}_0$; in other words, the $glF_1\times mK_1$ array $\mathbf{P}_1$ is defined as
	\begin{eqnarray}
		\label{eq-P-Q_0}
		\mathbf{P}_1=
		\left(
		\begin{array}{c}
			\mathbf{Q}_0\\
			\mathbf{Q}_0+S_1\\
			\vdots\\
			\mathbf{Q}_0+(gl-1)S_1
		\end{array}
		\right)
	\end{eqnarray}
	where $l=\frac{m}{\text{gcd}(L-m,m)}$.\footnote{\label{foot;value of l}
	By selecting this  value of $l$,  the vertical replication step in the construction of $\mathbf{P}_2$ contains a replication of integer times.
	More details will be explained later.}
	$\mathbf{P}_1$ contains totally $glS_1$ different integers, and each column of  $\mathbf{P}_1$  has $glZ_1$ stars. This is because, each array $\mathbf{Q}_0+jS_1$ where $j\in [0:gl-1]$ contains  $S_1$ different integers and each of its columns has $Z_1$ stars. Thus Conditions C$1$ and C$2$ of Definition \ref{def-MAPDA} hold.  Furthermore, Conditions C$3$ and C$4$ of Definition \ref{def-MAPDA} hold since each array $\mathbf{Q}_0+jS_1$ satisfies C$3$ and C$4$. So $\mathbf{P}_1$ is an $mg$-$(L,mK_1,glF_1,glZ_1,glS_1)$ MAPDA, with the set of integers  equal to  $[glF_1]$ where each integer occurs $mg$ times.

	When $L=3$ and $m=2$, we have $L-m=1$ and $l=\frac{m}{\text{gcd}(L-m,m)}=\frac{2}{\text{gcd}(1,2)}=2$. By replicating the $m$-$(K_1,F_1,Z_1,S_1)=3$-$(4,6,3,4)$ PDA $\mathbf{Q}$ listed in Fig. \ref{fig-ex1}   twice horizontally, we have the $6\times 8$ array $\mathbf{Q}_0$. Since $gl=6$, we obtain
	\begin{eqnarray*}
	\mathbf{P}_1=
	\left(
	\begin{array}{c}
		\mathbf{Q}_0\\
		\mathbf{Q}_0+4\\
		\mathbf{Q}_0+8\\
		\mathbf{Q}_0+12\\
		\mathbf{Q}_0+16\\
		\mathbf{Q}_0+20
	\end{array}
	\right),
	\end{eqnarray*}   as shown in Fig. \ref{fig-ex1}.

	\subsection{Construction of an MAPDA $\mathbf{P}_2$}\label{sub-P2}
	In the following we will construct another MAPDA $\mathbf{P}_2$ with the set of integers $[glS_1]$, where each integer appears exactly $L-m$ times.
	Hence, if we concatenate  $\mathbf{P}_1$ and  $\mathbf{P}_2$ vertically, each integer occurs $mg+L-m=L+m(g-1)$ times, coinciding with the achieved sum-DoF in Theorem~\ref{th-MAPDA-regular}.
	The main idea of constructing $\mathbf{P}_2$ is that we first replicate $\mathbf{Q}$ vertically with   integer right-shifting to obtain a new array $\mathbf{U}$, then replicate it horizontally to obtain a new array $\mathbf{U}_0$, and finally increase some integers in $\mathbf{U}_0$ to obtain $\mathbf{P}_2$.
	\subsubsection{Vertical replication with   integer right-shifting}
	\label{subsub-vertial}
	With the choice $l=\frac{m}{\text{gcd}(L-m,m)}$, we have that $\frac{l(L-m)}{m}$ is an integer.
	Assume that \begin{eqnarray*}
	\mathbf{Q}=
	\left(
	\begin{array}{c}
		\mathbf{q}_1\\
		\mathbf{q}_2\\
		\vdots\\
		\mathbf{q}_{F_1}
	\end{array}
	\right),
	\end{eqnarray*}
	where $\mathbf{q}_j$ represents the $j^{\text{th}}$ row of $\mathbf{Q}$, for each $j\in [F_1]$. We replicate $\mathbf{q}_j$, $\frac{l(L-m)}{m}$ times vertically, to obtain \begin{eqnarray*}
	\mathbf{A}_j'=
	\left(
	\begin{array}{c}
		\mathbf{q}_j\\
		\mathbf{q}_j\\
		\vdots\\
		\mathbf{q}_j
	\end{array}
	\right),
	\end{eqnarray*}
	with dimension $\frac{l(L-m)}{m} \times K_1$.
	Then we  obtain an $\frac{l(L-m)}{m}F_1\times K_1$ array $\mathbf{U}'$ by concatenating $\mathbf{A}_1', \mathbf{A}_2',\ldots, \mathbf{A}_{F_1}'$ vertically; i.e.,
	\begin{eqnarray*}
	\mathbf{U}'=
	\left(
	\begin{array}{c}
		\mathbf{A}_1'\\
		\mathbf{A}_2'\\
		\vdots\\
		\mathbf{A}_{F_1}'
	\end{array}
	\right).
	\end{eqnarray*}
	We can see that in $\mathbf{U}'$ each integer occurs $g\frac{l(L-m)}{m}$ times. In addition, each integer in each column of $\mathbf{U}'$    occurs $\frac{l(L-m)}{m}$ times, which violates Condition C$3$ of Definition \ref{def-MAPDA}.
	So we need to adjust the integers in $\mathbf{U}'$ to meet Condition C$3$ of Definition \ref{def-MAPDA}.

	We then generate a new array   $\mathbf{A}_j$ with dimension $\frac{l(L-m)}{m} \times K_1$,  by cyclically right-shifting the integers in $i^{\text{th}}$ row of $\mathbf{A}_j$ by $i$ positions while the  stars in this row remain unchanged, for all $i\in \left[\frac{l(L-m)}{m}\right]$.

	Since for any two different rows of $\mathbf{A}'_j$, we right-shift the same integer by different positions, then the following property is obtained.

	{\bf Property 1:} Each integer in $\mathbf{A}_j$ must occur in $\frac{l(L-m)}{m}$ different columns of $\mathbf{A}_j$. 		\hfill $\square$

	Next we generate an array $\mathbf{U}$ by vertically concatenating  $\mathbf{A}_1, \mathbf{A}_2, \ldots, \mathbf{A}_{F_1}$, whose dimension is $\frac{l(L-m)}{m}F_1\times K_1$; i.e.,
	\begin{eqnarray*}
	\mathbf{U}=
	\left(
	\begin{array}{c}
		\mathbf{A}_1\\
		\mathbf{A}_2\\
		\vdots\\
		\mathbf{A}_{F_1}
	\end{array}
	\right).
	\end{eqnarray*}
	Since $\mathbf{U}$ is obtained from $\mathbf{U}'$ by some right-shifting steps, each integer in $\mathbf{U}$ occurs the same times (i.e., $g\frac{l(L-m)}{m}$ times) as in $\mathbf{U}'$.

	Let us return to Example \ref{ex-3}. Recall that $\mathbf{Q}$ is a $3$-$(4,6,3,4)$ PDA and $L=3$, $m=2$ and $l=2$. We have $\frac{l(L-m)}{m}=\frac{2\times 1}{2}=1$.
	So in this example, we have   $\mathbf{U}'=\mathbf{Q}$.
 	Let us take $\mathbf{A}'_1=(*\ *\ 1\ 2)$, which is the first row of $\mathbf{U}'$,  to show our shifting method. We cyclically right-shift integers $1$ and $2$ to the next integer positions $(1,4)$ and $(1,3)$ respectively, to obtain the subarray $\mathbf{A}_1=(*\ *\ 2\ 1)$. Similarly we shift all the integers in other rows and obtain
	\begin{eqnarray}\label{array-U}
	\mathbf{U}=
	\left(
	\begin{array}{cccc}
		* & * & 2 & 1 \\
		* & 3 & * & 1 \\
		* & 3 & 2 & * \\
		4 & * & * & 1 \\
		4 & * & 2 & * \\
		4 & 3 & * & *
	\end{array}
	\right) =\left(
	\begin{array}{c}
		\mathbf{A}_1\\
		\mathbf{A}_2\\
		\vdots\\
		\mathbf{A}_{6}
	\end{array}\right),
	\end{eqnarray} as shown in Fig. \ref{fig-ex1}.
	It can be seen that each integer appears once and in one column in $\mathbf{A}_j$ where $j\in[6]$.

	\subsubsection{Horizontal replication}\label{subsub-horizontal}
	We then replicate $\mathbf{U}$   $m$ times horizontally to get an   array $\mathbf{U}_0$, whose dimension is $\frac{l(L-m)}{m}F_1\times mK_1$.  It can be checked that each integer occurs   $mg\frac{l(L-m)}{m}=gl(L-m)$ times in $\mathbf{U}_0$.
	For each $j\in[F_1]$, we define $\mathbf{B}_j$ as  a subarray of   $\mathbf{U}_0$,  which is  composed of the rows indexed by $(j-1)\frac{l(L-m)}{m}+1,(j-1)\frac{l(L-m)}{m}+2,\ldots,j\frac{l(L-m)}{m}$ in   $\mathbf{U}_0$.
	Thus  $\mathbf{B}_j$ is with the form,
	\begin{align*}
		\mathbf{B}_j=\left( \underbrace{\mathbf{A}_j,\ldots,\mathbf{A}_j}_{m}\right),
	\end{align*}
	and  $\mathbf{U}_0$ is with the form,
	\begin{eqnarray*}
	\mathbf{U}_0=
	\left(
	\begin{array}{c}
		\mathbf{B}_1\\
		\mathbf{B}_2\\
		\vdots\\
		\mathbf{B}_{F_1}
	\end{array}
	\right).
	\end{eqnarray*}

	By {\bf Property 1} the following property can be obtained directly.

	{\bf Property 2:} Each integer in $\mathbf{B}_j$ occurs in $m\frac{l(L-m)}{m}=l(L-m)$ distinct columns of $\mathbf{B}_j$. 		\hfill $\square$

	Let us go back to Example \ref{ex-3}.   From \eqref{array-U},  we have
	\begin{eqnarray}\label{array-U0}
	\mathbf{U}_0=
	\left(
	\begin{array}{cccc|cccc}
		* & * & 2 & 1 & * & * & 2 & 1\\
		* & 3 & * & 1 & * & 3 & * & 1\\
		* & 3 & 2 & * & * & 3 & 2 & *\\
		4 & * & * & 1 & 4 & * & * & 1\\
		4 & * & 2 & * & 4 & * & 2 & *\\
		4 & 3 & * & * & 4 & 3 & * & *
	\end{array}
	\right) =\left(
	\begin{array}{c}
		\mathbf{B}_1\\
		\mathbf{B}_2\\
		\vdots\\
		\mathbf{B}_{6}
	\end{array}
	\right)=
	\left(
	\begin{array}{c|c}
		\mathbf{A}_1&\mathbf{A}_1\\
		\mathbf{A}_2&\mathbf{A}_2\\
		\vdots&\vdots\\
		\mathbf{A}_6&\mathbf{A}_6
	\end{array}
	\right),
	\end{eqnarray}
	as shown in $\mathbf{U}_0$ of Fig. \ref{fig-ex1}.
	The subarray
	$\mathbf{B}_1=(
	\begin{array}{cccccccc}
		*&*&2&1&*&*&2&1
	\end{array}
	)$
	is the first row  of $\mathbf{U}_0$, can be regraded as replicating
	$\mathbf{A}_1=
	(
	\begin{array}{cccc}
		*&*&2&1
	\end{array}
	)$
 	twice horizontally, and each integer appears $l(L-m)=2$ distinct columns in $\mathbf{B}_1$.

	\subsubsection{Increase of  the integers in $\mathbf{U}_0$}
	\label{subsub-arrange}
	In $\mathbf{U}_0$, the set of integers is $[S_1]$, while each integer occurs $gl(L-m)$ times. The objective of this step is to increase the integers in $\mathbf{U}_0$, such that the  set  of integers becomes $[glS_1]$ and each integer occurs $L-m$ times.

	For each integer  $s\in [S_1]$ in $\mathbf{U}_0$, we sort its  $gl(L-m)$ replicas   from left to right and top to bottom in $\mathbf{U}_0$.
	We increase the $i^{\text{th}}$ replica  where $i\in [gl(L-m)]$ by $\lfloor\frac{i-1}{L-m}\rfloor S_1$; i.e., the $i^{\text{th}}$ replica now becomes $s+\lfloor\frac{i-1}{L-m}\rfloor S_1$. Hence, all the replicas  with the order in $[(j-1)(L-m)+1:j(L-m)]$ are replaced by integer $s+(j-1) S_1$, where $j\in [gl]$.
 	Thus we introduce $gl$ integers to replace the replicas of  $s$  in $\mathbf{U}_0$, where each integer occurs $L-m$  times.

	After considering all integers  $s\in [S_1]$ and increasing the replicas of the integers  as described above, the resulting array is  $\mathbf{P}_2$, which totally contains $glS_1$ different integers. In other words, the set of integers in   $\mathbf{P}_2$ is $[glS_1]$, which    is the same as that in $\mathbf{P}_1$.   Thus Condition C$2$ of Definition \ref{def-MAPDA} holds. Furthermore, recall that in our vertical replication step the number of times to replicate $\mathbf{Q}$ is $\frac{(L-m)l}{m}$,  and that the right-shifting step does not change the positions of stars.
	Hence, there are $\frac{(L-m)l}{m}Z_1$ stars in each column of $\mathbf{P}_2$. Then Condition C$1$ of Definition \ref{def-MAPDA} holds.

	Now let us consider Condition C$3$ of Definition \ref{def-MAPDA}. Let $\mathbf{D}_j$ be the row-wise subarray of $\mathbf{P}_2$, which is composed of the rows indexed by $(j-1)\frac{l(L-m)}{m}+1,(j-1)\frac{l(L-m)}{m}+2,\ldots, j\frac{l(L-m)}{m}$. Then $\mathbf{P}_2$ can be represented as
	\begin{eqnarray*}
	\mathbf{P}_2=
	\left(
		\begin{array}{c}
			\mathbf{D}_1\\
			\mathbf{D}_2\\
			\vdots\\
			\mathbf{D}_{F_1}
		\end{array}
	\right).
	\end{eqnarray*}
 	For each $j\in [F_1]$, since $\mathbf{D}_j$ is obtained by increasing the integers in $\mathbf{B}_j$, and each integer in $\mathbf{B}_j$ occurs   $l(L-m)$ times in $\mathbf{B}_j$, then each integer in $\mathbf{D}_j$ occurs   $L-m$ times in $\mathbf{D}_j$.
	This is because, to obtain $\mathbf{D}_j$ from $\mathbf{B}_j$, we replace the $l(L-m)$  replicas of each integer in  $\mathbf{B}_j$ by
	$l$ different integers in $\mathbf{D}_j$; in other words, each  consecutive $L-m$ replicas  from left to right and top to bottom of each integer in  $\mathbf{B}_j$ are replaced by the same integer in $\mathbf{D}_j$.

	Recall that  each integer in $\mathbf{P}_2$ occurs $L-m$ times in $\mathbf{P}_2$. Hence, each integer in $\mathbf{D}_j$ does not occur in  $\mathbf{D}_{j_1}$ where $j_1 \in [F_1]\setminus\{j\}$. Together with   {\bf Property 2}, it can be seen that each integer in $\mathbf{D}_j$ occurs
    in $L-m$ distinct columns of $\mathbf{D}_j$ . Then Condition C$3$ of Definition \ref{def-MAPDA} holds.

 	Finally let us consider Condition C$4$ of Definition \ref{def-MAPDA}. Since the subarray $\mathbf{P}_2^{(s)}$ of $\mathbf{P}_2$ including the rows and columns containing $s\in[glS_1]$ has $L-m$ columns which is less than $L$, then C$4$ holds. Thus $\mathbf{P}_2$ is an $(L-m)$-$(L$, $mK_1$, $\frac{(L-m)l}{m}F_1$, $\frac{(L-m)l}{m}Z_1$, $glS_1)$ MAPDA.

	Let us see Example \ref{ex-3} again. From \eqref{array-U0}, each integer appears  $gl(L-m)=6$ times in $\mathbf{U}_0$. Since $L-m=1$ we replace integer $s=1$ in positions $(1,4)$, $(1,8)$, $(2,4)$, $(2,8)$, $(4,4)$, $(4,8)$ of $\mathbf{U}_0$ by
	$$
	\begin{array}{ll}
		1+\left\lfloor\frac{1-1}{1} \right\rfloor \times 4=1+0\times 4=1,&
		1+\left\lfloor\frac{2-1}{1} \right\rfloor \times 4=1+1\times 4=5,\\
		1+\left\lfloor\frac{3-1}{1} \right\rfloor \times 4=1+2\times 4=9,&
		1+\left\lfloor\frac{4-1}{1} \right\rfloor \times 4=1+3\times 4=13,\\
		1+\left\lfloor\frac{5-1}{1} \right\rfloor \times 4=1+4\times 4=17,&
		1+\left\lfloor\frac{6-1}{1} \right\rfloor \times 4=1+4\times4=21.
	\end{array}
	$$ respectively. Then we have
	$$
	\begin{array}{lll}
		\mathbf{P}_2(1,4)=1,& \mathbf{P}_2(1,8)=5,&  \mathbf{P}_2(2,4)=9,\\
		\mathbf{P}_2(2,8)=13, & \mathbf{P}_2(4,4)=17,&  \mathbf{P}_2(4,8)=21.
	\end{array}
	$$ Similarly after replacing all the other integers $s=2$, $3$ and $4$, we get
	\begin{eqnarray}\label{array-P2}
	\mathbf{P}_2=
	\left(
	\begin{array}{cccc|cccc}
		*    & *    & 2    & 1    & *    & *    & 6  & 5 \\
		*    & 3    & *    & 9  & *    & 7  & *    & 13\\
		*    & 11  & 10 & *    & *    & 15 & 14 & *   \\
		4    & *    & *    & 17 & 8  & *    & *    & 21\\
		12  & *    & 18 & *    & 16 & *    & 22 & *   \\
		20 & 19 & *    & *    & 24 & 23 & *    & *
	\end{array}
	\right)
	=\left(
		\begin{array}{c}
			\mathbf{D}_1\\
			\mathbf{D}_2\\
			\vdots\\
			\mathbf{D}_{6}
		\end{array}\right)
	\end{eqnarray}
	which is shown in Fig. \ref{fig-ex1}. The subarray
	$\mathbf{D}_1=(
	\begin{array}{cccccccc}
		*& *& 2& 1& *& *& 6  & 5
	\end{array}
	)$
	is the first row  of $\mathbf{P}_2$, can be regraded as increasing the integers in
	$\mathbf{B}_1=
	(
	\begin{array}{cccccccc}
		*&*&2&1&*&*&2&1
	\end{array}
	)$, and each integer only appears in one column of $\mathbf{D}_1$.We can see that each integer occurs exactly once in $\mathbf{P}_2$.

	\subsection{Construction of an MAPDA $\mathbf{P}$}
	The last step is to obtain an array $\mathbf{P}$ by concatenating $\mathbf{P}_1$ and $\mathbf{P}_2$ vertically, i.e., $\mathbf{P}=[\mathbf{P}_1;\mathbf{P}_2]$. The  set of integers in $\mathbf{P}$ is $[glS_1]$. The numbers of rows and columns in $\mathbf{P}$ are $F=(g+\frac{L-m}{m})lF_1=\alpha F_1$ and $K=mK_1$, respectively, where $\alpha=(g+\frac{L-m}{m})l$ as defined in Theorem \ref{th-MAPDA-regular}.

	Next we will show that $\mathbf{P}$ is an $(L+m(g-1))$-$(L,mK_1$, $\alpha F_1$, $\alpha Z_1$, $glS_1)$ MAPDA. Since there are $glZ_1$ and $\frac{l(L-m)}{m}Z_1$ stars in each column of $\mathbf{P}_1$ and $\mathbf{P}_2$ respectively, then the number of stars in each column of $\mathbf{P}$ is $Z=(g+\frac{L-m}{m})lZ_1=\alpha Z_1$. Thus Condition C$1$ of Definition \ref{def-MAPDA} holds. In other words, each user caches $\frac{ ZN}{F}=\frac{\alpha Z_1N}{\alpha F_1}=\frac{Z_1N}{F_1}=M$ files, which satisfies the memory size constraint. Furthermore, since each integer $s\in [glS_1]$ appears $mg$ times and $L-m$ times in $\mathbf{P}_1$ and $\mathbf{P}_2$ respectively, then each integer occurs in $\mathbf{P}$ exactly $mg+L-m=m(g-1)+L$ times. So Condition C$2$ of Definition \ref{def-MAPDA} holds.

	In order to verify the Conditions C3 and C4 of Definition \ref{def-MAPDA}, the  following Lemma \ref{lem-different-column-times}  is useful, whose proof is included in Appendix \ref{proof-lem-4}.
	\begin{lemma}
		\label{lem-different-column-times}
		By our construction, the following statements hold:
	\begin{itemize}	
	\item Each integer $s\in[glS_1]$ occurs in $mg+L-m=m(g-1)+L$ distinct columns in $\mathbf{P}$.
	\item
 	In $\mathbf{P}^{(s)}_1$, which is the subarray of  $\mathbf{P}_1$ including the rows and columns containing $s$, each row has exactly $m(g-1)$ stars.
	\item For any row of $\mathbf{P}_2$ containing $s$, there must exist a row of $\mathbf{P}_1$ containing $s$, where  the   stars are located at the same positions in these two rows.
	\end{itemize}
		\hfill $\square$
	\end{lemma}

	From the first statement of Lemma \ref{lem-different-column-times}, $\mathbf{P}$ satisfies Condition C$3$ of Definition \ref{def-MAPDA}.

	Finally, let us consider Condition C$4$ of Definition \ref{def-MAPDA}. Denote the column index sets of columns containing $s$ of $\mathbf{P}$, $\mathbf{P}_1$ and $\mathbf{P}_2$ by $\mathcal{K}^{(s)}$, $\mathcal{K}^{(s)}_1$ and, $\mathcal{K}^{(s)}_2$ respectively. By the first statement of Lemma \ref{lem-different-column-times}, the subarray $\mathbf{P}^{(s)}$ including the rows and columns of $\mathbf{P}$ containing $s$ has $m(g-1)+L$ columns. This implies that $\mathcal{K}^{(s)}=\mathcal{K}^{(s)}_1\cup \mathcal{K}^{(s)}_2$ and $\mathcal{K}^{(s)}_1\cap \mathcal{K}^{(s)}_2=\emptyset$. By the second and third statements of  Lemma \ref{lem-different-column-times}, the number of stars in each row of $\mathbf{P}^{(s)}$ is at least $m(g-1)$. So the number of integer entries in each row of $\mathbf{P}^{(s)}$ is at most
	$L+m(g-1)-m(g-1)=L$. Then Condition C$4$ of Definition \ref{def-MAPDA} holds.

 	From the above discussions, when $m<L$, $\mathbf{P}$ is an ($m(g-1)+L)$-$(L,K=mK_1,\alpha F_1,\alpha Z_1,glS_1)$ MAPDA and  $\text{sgn}(g)=g$.
 	In conclusion,  the proof of Theorem~\ref{th-MAPDA-regular} is completed.
			
	\section{Conclusion}\label{sec-conclu}
	In this paper, we studied  the cache-aided MISO broadcast channel  problem with one-shot linear delivery. We first presented a  new design construction, referred to as  MAPDA,  to characterize the placement and delivery phases.
	For the system with parameters satisfying $KM/N+L=K$, we proposed a  scheme  under MAPDA  to achieve the maximum sum-DoF with subpacketization equals to $K$. For the general case, we proposed another scheme by   a  non-trivial transformation approach from any regular PDA for the original caching problem. If the original PDA is the MN PDA, the resulting scheme can achieve maximum sum-DoF with lower subpacketization than the existing schemes.

	\begin{appendices}
	\section{proof of Theorem \ref{th-DoF}}\label{proof-DoF}
		
	\begin{proof}
		Assume that each integer $s\in[S]$ occurs $r_s$ times in $\mathbf{P}$, denoted by $\mathbf{P}(f_1,k_1)$, $\mathbf{P}(f_2,k_2)$, $\ldots$, $\mathbf{P}(f_{r_s},k_{r_s})$. We can obtain the subarray $\mathbf{P}^{(s)}$ with $r_s$ columns from Condition C$3$ of Definition \ref{def-MAPDA}, i.e., each integer occurs in each column at most once,
		and  let $f_i$ and $k_i$, $i\in [r_s]$ represent the row indices and column indices of $\mathbf{P}^{(s)}$, respectively. For each subarray $\mathbf{P}^{(s)}$, $s\in[S]$, we assume that there are $r_{s,i}$ integer entries in the row $f_i$. Then the number of stars used by all the integer $s$'s in $\mathbf{P}^{(s)}$ is exactly $\sum_{i=1}^{r_s}(r_s-r_{s,i})$, and the total number of stars used in all $\mathbf{P}^{(s)}$, $s\in [S]$, is
		\begin{eqnarray*}
			M=\sum_{s=1}^S\sum_{i=1}^{r_s}(r_s-r_{s,i})=\sum_{s=1}^Sr_s^2-\sum_{s=1}^S\sum_{i=1}^{r_s}r_{s,i}.
		\end{eqnarray*}
		Next, we consider the array $\mathbf{P}$ and assume that each row $j\in [F]$ has $r'_j$ integer entries, then the times of all stars used by the integer entries in $j^{\text{th}}$ row is at most $r'_j(K-r'_j)$. So the total times of all stars used in $\mathbf{P}$ is at most $M'=\sum_{j=1}^Fr'_j(K-r'_j)$. Clearly, $M \leq M'$, i.e.,
		\begin{eqnarray}\label{num-star}
			\sum_{s=1}^Sr_s^2-\sum_{s=1}^S\sum_{i=1}^{r_s}r_{s,i} \leq \sum_{j=1}^Fr'_j(K-r'_j)\label{1}.
		\end{eqnarray}
		Since $n=(F-Z)K$ is the total number of integers in $\mathbf{P}$, we have $n=\sum_{s=1}^Sr_s =\sum_{i=j}^F r'_j$. From \eqref{num-star}, we get
		\begin{eqnarray*}
			\sum_{s=1}^{S}(r_s)^2 +\sum_{j=1}^{F}(r'_j)^2 \leq \sum_{s=1}^{S}\sum_{i=1}^{r_s}r_{s,i} +\sum_{j=1}^{F}Kr'_j.
		\end{eqnarray*}
		In addition, by the convexity and
		$\sum_{i=1}^{r_s}r_{s,i}\leq r_sL$ from Condition C$4$ of Definition \ref{def-MAPDA}, we can obtain
		$$\sum_{s=1}^{S}(r_s)^2\geq \frac{1}{S}(\sum_{s=1}^{S}r_s)^2=\frac{n^2}{S},\ \ \ \ \sum_{s=1}^{F}(r'_j)^2\geq \frac{1}{F}(\sum_{s=1}^{S}r'_j)^2=\frac{n^2}{F}.$$
		Then
		\begin{eqnarray*}
			\frac{n^2}{S}+\frac{n^2}{F}\leq \sum_{s=1}^{S}r_sL+Kn = nL+Kn,
		\end{eqnarray*} i.e., $S\geq \frac{nF}{FL+KF-n}$.
		Thus we get the sum-DoF of $\frac{K(F-Z)}{S}\leq \frac{FL+KZ}{F}=\frac{KZ}{F}+L$, where the equation holds if and only if $r_{s,i}=L$, $r_1=r_2=\cdots=r_{S}=\frac{n}{S}$ and $r'_1=r'_2=\cdots=r'_F$. Then the proof is completed.
	\end{proof}

	\section{proof of Lemma \ref{lem-different-column-times}}\label{proof-lem-4}
	Let us consider the second statement first. From \eqref{eq-P-Q_0}, the rows of $\mathbf{P}_1$ containing $s$ are exactly the rows of the array $\mathbf{Q}_0+jS_1$ for some integer $j\in [0:gl-1]$. Since $\mathbf{Q}_0+jS_1$ is an $mg$-$(L,mK_1,F_1,Z_1,S_1)$ MAPDA, each row of $\mathbf{P}^{(s)}_1$ which is generated by the rows and columns of $\mathbf{P}_1$ containing $s$ has exactly $m(g-1)$ stars. So the second statement holds. Furthermore $\mathbf{P}^{(s)}_1$ has the following form
	\begin{eqnarray}
	\label{eq-P^s_1}
	\mathbf{P}_1^{(s)}=
	\left(
		\begin{array}{cccc|c|cccc}
			s      &*      &\cdots &*	   &\cdots &s      &*      &\cdots &*      \\
			*      &s      &\cdots &*	   &\cdots &*      &s      &\cdots &*      \\
			\vdots &\vdots &\ddots &\vdots &\cdots &\vdots &\vdots &\ddots &\vdots \\
			*      &*      &\cdots &s      &\cdots &*      &*      &\cdots &s      \\
		\end{array}
	\right)_{mg\times g}
	\end{eqnarray} with the row and column permutations.

	Now let us consider the third statement. For any integer $s\in[glS_1]$, from the construction of $\mathbf{P}_1$ in Section \ref{sub-P1} and the construction of $\mathbf{P}_2$ in Section \ref{sub-P2}, the integer $s$ at row $\mathbf{p}_{1,j_1}$ of $\mathbf{P}_1$ and at row $\mathbf{p}_{2,j_2}$ of $\mathbf{P}_2$ can be written as follows respectively where $j_1\in [glF_1]$ and $j_2\in [\frac{l(L-m)}{m}F_1]$.
	\begin{eqnarray*}
		s=s'+(h-1)S_1,\ \ \  s=s''+\left\lfloor\frac{i-1}{L-m}\right\rfloor S_1,\ \ \ h\in[gl],\ i\in[(L-m) gl],\ s',s''\in[S_1].
	\end{eqnarray*}Without loss of generality, we assume that $s'\geq s''$. Then we have
	$$S_1>s'-s''=\left\lfloor\frac{i-1}{L-m}\right\rfloor S_1-(h-1)S_1=\left(\left\lfloor\frac{i-1}{L-m}\right\rfloor-(h-1)\right)S_1.$$
	The above equality holds if and only if $s'=s''$ and $\left\lfloor\frac{i-1}{L-m}\right\rfloor=h-1$ hold. This implies that the row $\mathbf{p}_{2,j_2}$ have the same star positions as that of some rows $\mathbf{p}_{1,j_1}$. Then the second statement holds.

	Finally let us consider the first statement. Recall that $\mathbf{P}_2$ is generated through the step of cyclically-right-shifting the integers into the other integer positions and remaining the star positions in each row of $\mathbf{Q}$. From \eqref{eq-P^s_1} and second statement, the indices of the columns containing $s$ of $\mathbf{P}_2$ must different from the indices of the columns containing $s$ of $\mathbf{P}_1$. Then the first statement holds.
\end{appendices}

\bibliographystyle{IEEEtran}
\bibliography{reference}

\begin{thebibliography}{10}
\providecommand{\url}[1]{#1}
\csname url@samestyle\endcsname
\providecommand{\newblock}{\relax}
\providecommand{\bibinfo}[2]{#2}
\providecommand{\BIBentrySTDinterwordspacing}{\spaceskip=0pt\relax}
\providecommand{\BIBentryALTinterwordstretchfactor}{4}
\providecommand{\BIBentryALTinterwordspacing}{\spaceskip=\fontdimen2\font plus
\BIBentryALTinterwordstretchfactor\fontdimen3\font minus
  \fontdimen4\font\relax}
\providecommand{\BIBforeignlanguage}[2]{{%
\expandafter\ifx\csname l@#1\endcsname\relax
\typeout{** WARNING: IEEEtran.bst: No hyphenation pattern has been}%
\typeout{** loaded for the language `#1'. Using the pattern for}%
\typeout{** the default language instead.}%
\else
\language=\csname l@#1\endcsname
\fi
#2}}
\providecommand{\BIBdecl}{\relax}
\BIBdecl

\bibitem{MN}
M.~A. Maddah-Ali and U.~Niesen, ``Fundamental limits of caching,'' \emph{IEEE
  Transactions on Information Theory}, vol.~60, no.~5, pp. 2856--2867, 2014.

\bibitem{YCTC}
Q.~Yan, M.~Cheng, X.~Tang, and Q.~Chen, ``On the placement delivery array
  design for centralized coded caching scheme,'' \emph{IEEE Transactions on
  Information Theory}, vol.~63, no.~9, pp. 5821--5833, 2017.

\bibitem{CJYT}
M.~Cheng, J.~Jiang, Q.~Yan, and X.~Tang, ``Constructions of coded caching
  schemes with flexible memory size,'' \emph{IEEE Transactions on
  Communications}, vol.~67, no.~6, pp. 4166--4176, 2019.

\bibitem{CJWY}
M.~Cheng, J.~Jiang, Q.~Wang, and Y.~Yao, ``A generalized grouping scheme in
  coded caching,'' \emph{IEEE Transactions on Communications}, vol.~67, no.~5,
  pp. 3422--3430, 2019.

\bibitem{CWZW}
M.~Cheng, J.~Wang, X.~Zhong, and Q.~Wang, ``A framework of constructing
  placement delivery arrays for centralized coded caching,'' \emph{IEEE
  Transactions on Information Theory}, vol.~67, no.~11, pp. 7121--7131, 2021.

\bibitem{WCWC}
J.~Wang, M.~Cheng, K.~Wan, and G.~Caire, ``Novel frameworks for coded caching
  via cartesian product with reduced subpacketization,'' \emph{arXiv preprint
  arXiv:2108.08486}, 2021.

\bibitem{TR}
L.~Tang and A.~Ramamoorthy, ``Coded caching schemes with reduced
  subpacketization from linear block codes,'' \emph{IEEE Transactions on
  Information Theory}, vol.~64, no.~4, pp. 3099--3120, 2018.

\bibitem{SZG}
C.~Shangguan, Y.~Zhang, and G.~Ge, ``Centralized coded caching schemes: A
  hypergraph theoretical approach,'' \emph{IEEE Transactions on Information
  Theory}, vol.~64, no.~8, pp. 5755--5766, 2018.

\bibitem{STD}
K.~Shanmugam, A.~M. Tulino, and A.~G. Dimakis, ``Coded caching with linear
  subpacketization is possible using ruzsa-szem{\'e}redi graphs,'' in
  \emph{2017 IEEE International Symposium on Information Theory (ISIT)}, 2017,
  pp. 1237--1241.

\bibitem{YTCC}
Q.~Yan, X.~Tang, Q.~Chen, and M.~Cheng, ``Placement delivery array design
  through strong edge coloring of bipartite graphs,'' \emph{IEEE Communications
  Letters}, vol.~22, no.~2, pp. 236--239, 2018.

\bibitem{KP}
P.~Krishnan, ``Coded caching via line graphs of bipartite graphs,'' in
  \emph{2018 IEEE Information Theory Workshop (ITW)}, 2018, pp. 1--5.

\bibitem{CLTW}
M.~Cheng, J.~Li, X.~Tang, and R.~Wei, ``Linear coded caching scheme for
  centralized networks,'' \emph{IEEE Transactions on Information Theory},
  vol.~67, no.~3, pp. 1732--1742, 2021.

\bibitem{JCM}
M.~Ji, G.~Caire, and A.~F. Molisch, ``Fundamental limits of caching in wireless
  {D2D} networks,'' \emph{IEEE Transactions on Information Theory}, vol.~62,
  no.~2, pp. 849--869, 2016.

\bibitem{KNMD}
N.~Karamchandani, U.~Niesen, M.~A. Maddah-Ali, and S.~N. Diggavi,
  ``Hierarchical coded caching,'' \emph{IEEE Transactions on Information
  Theory}, vol.~62, no.~6, pp. 3212--3229, 2016.

\bibitem{JWTCEL}
M.~Ji, M.~F. Wong, A.~M. Tulino, J.~Llorca, G.~Caire, M.~Effros, and
  M.~Langberg, ``On the fundamental limits of caching in combination
  networks,'' in \emph{2015 IEEE 16th International Workshop on Signal
  Processing Advances in Wireless Communications (SPAWC)}, 2015, pp. 695--699.

\bibitem{SSB}
S.~P. Shariatpanahi, S.~A. Motahari, and B.~H. Khalaj, ``Multi-server coded
  caching,'' \emph{IEEE Transactions on Information Theory}, vol.~62, no.~12,
  pp. 7253--7271, 2016.

\bibitem{NMA}
N.~Naderializadeh, M.~A. Maddah-Ali, and A.~S. Avestimehr, ``Fundamental limits
  of cache-aided interference management,'' \emph{IEEE Transactions on
  Information Theory}, vol.~63, no.~5, pp. 3092--3107, 2017.

\bibitem{HND}
J.~Hachem, U.~Niesen, and S.~N. Diggavi, ``Degrees of freedom of cache-aided
  wireless interference networks,'' \emph{IEEE Transactions on Information
  Theory}, vol.~64, no.~7, pp. 5359--5380, 2018.

\bibitem{SCH}
S.~P. Shariatpanahi, G.~Caire, and B.~Hossein~Khalaj, ``Physical-layer schemes
  for wireless coded caching,'' \emph{IEEE Transactions on Information Theory},
  vol.~65, no.~5, pp. 2792--2807, 2019.

\bibitem{EP}
E.~Lampiris and P.~Elia, ``Adding transmitters dramatically boosts
  coded-caching gains for finite file sizes,'' \emph{IEEE Journal on Selected
  Areas in Communications}, vol.~36, no.~6, pp. 1176--1188, 2018.

\bibitem{SPSET}
M.~J. Salehi, E.~Parrinello, S.~P. Shariatpanahi, P.~Elia, and A.~T{\"o}lli,
  ``Low-complexity high-performance cyclic caching for large miso systems,''
  \emph{IEEE Transactions on Wireless Communications}, pp. 1--1, 2021.

\bibitem{MB}
S.~Mohajer and I.~Bergel, ``{MISO} cache-aided communication with reduced
  subpacketization,'' in \emph{2020 IEEE International Conference on
  Communications (ICC)}, 2020, pp. 1--6.

\bibitem{ST}
M.~Salehi and A.~T{\"o}lli, ``Diagonal multi-antenna coded caching for reduced
  subpacketization,'' in \emph{GLOBECOM 2020 - 2020 IEEE Global Communications
  Conference}, 2020, pp. 1--6.

\bibitem{STSK}
M.~Salehi, A.~T{\"o}lli, S.~P. Shariatpanahi, and J.~Kaleva,
  ``Subpacketization-rate trade-off in multi-antenna coded caching,'' in
  \emph{2019 IEEE Global Communications Conference (GLOBECOM)}, 2019, pp. 1--6.

\bibitem{PJC}
E.~Piovano, H.~Joudeh, and B.~Clerckx, ``Generalized degrees of freedom of the
  symmetric cache-aided miso broadcast channel with partial csit,'' \emph{IEEE
  Transactions on Information Theory}, vol.~65, no.~9, pp. 5799--5815, 2019.

\bibitem{EBPresolving}
E.~Lampiris, A.~Bazco-Nogueras, and P.~Elia, ``Resolving the feedback
  bottleneck of multi-antenna coded caching,'' \emph{IEEE Transactions on
  Information Theory}, Dec. 2021.

\bibitem{CJ}
J.~W. Cotton, ``Latin square designs,'' \emph{Applied analysis of variance in
  behavioral science}, vol. 137, pp. 147--196, 1993.

\bibitem{1705002}
T.~Ho, M.~Medard, R.~Koetter, D.~Karger, M.~Effros, J.~Shi, and B.~Leong, ``A
  random linear network coding approach to multicast,'' \emph{IEEE Transactions
  on Information Theory}, vol.~52, no.~10, pp. 4413--4430, 2006.

\end{thebibliography}

\end{document}